\documentclass[
	paper=A4,
	pagesize=auto,
	fontsize=12pt
]{scrartcl}
\usepackage{libertine}
\recalctypearea

\usepackage[ngerman,british]{babel}
\usepackage[style=alphabetic,url=false,isbn=false,doi=true,eprint=false,backend=biber]{biblatex}
\usepackage{graphicx}
\usepackage{amsmath,amssymb,amsthm}
\usepackage{csquotes}
\usepackage{enumitem}
\usepackage{bm}

\usepackage{hyperref}
\usepackage{xurl}
\hypersetup{colorlinks,breaklinks=true}

\usepackage{cleveref}

\newtheorem{theorem}{Theorem}[section]
\newtheorem{corollary}[theorem]{Corollary}
\newtheorem{lemma}[theorem]{Lemma}
\newtheorem{definition}[theorem]{Definition}
\newtheorem{remark}[theorem]{Remark}
\newtheorem{example}[theorem]{Example}

\newlength{\JZHeightOfX}
\newcommand{\JZOrcidlink}[1]{
\setlength{\JZHeightOfX}{\fontcharht\font`X}
\includegraphics[height=\JZHeightOfX]{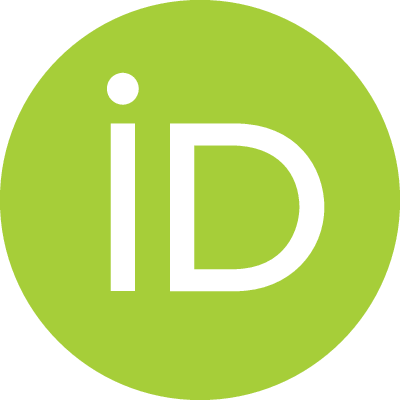}
\href{https://orcid.org/#1}{#1}
}

\newcommand{\tagx}[1][]{\refstepcounter{equation}(\theequation)\label{#1}}

\usepackage{leftidx}

\begin{document}
\title{A Rigorous Derivation of the Functional Renormalisation Group Equation}
\author{Jobst Ziebell\hspace{2em}\JZOrcidlink{0000-0002-9715-6356}\\
\small{Theoretisch-Physikalisches Institut, Friedrich-Schiller-University, Jena, Germany}}
\date{4 September 2023}

\maketitle

\begin{abstract}
The functional renormalisation group equation is derived in a mathematically rigorous fashion in a framework suitable for the Osterwalder-Schrader formulation of quantum field theory.
To this end, we devise a very general regularisation scheme which retains reflection positivity as well as the infinite degrees of freedom including smoothness.
Furthermore, it is shown how the classical limit is altered by the regularisation process leading to an inevitable breaking of translation invariance.
We also give precise conditions for the convergence of the obtained theories upon removal of the regularisation.
\end{abstract}

\section{Introduction}
It is still unknown whether interacting and non-perturbative quantum field theories encoding relevant physical phenomena can be given precise mathematical meaning.
Most prominently, it is unclear how to proceed from perturbative calculations to exact ones except in some special cases such as conformal field theories.
A much celebrated result that is frequently used is the functional renormalisation group equation \cite{src:Ellwanger1994,src:Morris1994} (FRGE) in the form presented by Wetterich \cite{src:Wetterich1993}.
It is generally accepted as an exact equation encoding the flow from a classical action functional to its quantum counterpart and many results have been derived or conjectured from its exact properties or practical approximations.
Some notable reviews include \cite{src:Delamotte2012,src:EichhornEtAl2020,src:Gies2012}.
However, the derivation of the FRGE has not yet been put into a mathematically satisfactory form but has remained restricted to formal respectively perturbative proofs \cite{src:Keller:PerturbativeFlows}.

In this work, we investigate the case of a real scalar field on Euclidean spacetime.
In the framework of the Osterwalder-Schrader theorem \cite{src:OsterwalderSchrader1, src:OsterwalderSchrader2} as presented by Glimm and Jaffe \cite{src:GlimmJaffe} it is modelled by a probability measure on the space of tempered distributions.
However, starting from a classical action it is not generally known how to define a corresponding measure.
Instead, it is customary to produce well-defined regularised theories and investigate the convergence of the corresponding measures upon the removal of the regularisation.
In this work we shall investiate how this procedure can be made to work in tandem with Wetterich's equation.

We begin by listing some basic tools in section \ref{sec:Preliminaries} that are necessary for the overall analysis.
In particular, we cover basic features of Radon Gaußian measures on locally convex spaces along with some standard results in convex analysis.
In section \ref{sec:LSCEnvelopes} we generalise the concept of lower semicontinuous envelopes and introduce the property of supercoercivity.
The latter enables a direct computation of a lower semicontinuous envelope.

Section \ref{sec:Convergence} then studies the convergence of measures on the space of tempered distributions in terms of convergence of dual objects that are fundamental to Wetterich's equation.
The connection is found to be given in terms of Mosco convergence and Attouch-Wets convergence, concepts which both enable a continuity theorem of the Legendre-Fenchel transform.

In section \ref{sec:FunctionalRegularisationScheme} we present a regularisation scheme tailored to the derivation of Wetterich's equation.
ts interplay with reflection positivity is studied in section \ref{sec:ReflectionPositivity} and finally, the proof of Wetterich's equation is given in section \ref{sec:WetterichDerivation}.

The derivation may be summarised as follows:
We define a family $Z_k$ of moment-generating functions (generating functionals) as well as their logarithms $W_k$ of a regularised theory on a Hilbert completion of the space of tempered distributions.
Then, we proceed by showing that the corresponding Fréchet derivatives $D W_k$ are linear bijections of Hilbert spaces such that the Legendre-Fenchel transforms (convex conjugates) $\Gamma_k$ of $W_k$ can be given in terms of $W_k$ and $(D W_k)^{-1}$.
It is then easily shown that $\Gamma_k$ - also referred to as the \enquote{effective average action} - satisfies a differential equation.
Furthermore, we compute the \enquote{classical limit} of $\Gamma_k$ which turns out to be the classical action modulo some contributions related to regularisation.

Finally, the full quantum field theory may be obtained by studying the convergence of the regularised effective average actions in terms of the result obtained in section \ref{sec:Convergence}.
However, it it not clear whether - or in which cases - Wetterich's equation itself survives the passage to the corresponding limit.
\section{Preliminaries and Conventions}
\label{sec:Preliminaries}
Every vector space in this work is taken to be real unless explicitly stated otherwise and the complexification of a real vector space $V$ will be denoted by $V_{\mathbb{C}}$.
\begin{definition}
\label{def:FourierTransform}
We define the Fourier transform $\hat{f}$ of a measurable function $f : \mathbb{R}^d \to \mathbb{R}$ as
\begin{equation}
\hat{f} \left( p \right) = \left( 2 \pi \right)^{-\frac{d}{2}} \int_{\mathbb{R}^d} \exp \left[ - i p x \right] f \left( x \right) \mathrm{d} x \, ,
\end{equation}
whenever the integral converges.
The corresponding unitary operator on $L^2(\mathbb{R}^d)_{\mathbb{C}}$ is denoted by $\mathcal{F}$ and the function $x \to \exp [ - i p x ] / ( 2 \pi )^{d/2}$ interpreted as a tempered distribution is denoted by $\mathcal{F}_p$.
\end{definition}

We shall work with the topological vector space $\mathcal{S} ( \mathbb{R}^d )$ of Schwartz functions over $\mathbb{R}^d$ for $d \in \mathbb{N}$ which for brevity we shall refer to as $\mathcal{S}$.
The corresponding space of tempered distributions with its strong dual topology will be referred to as $\mathcal{S}'_\beta$.
We shall generally let $X'$ denote the dual space of a locally convex space $X$ and reserve the notion $X'_\beta$ to encode the dual space with its strong dual topology.
In the case of normed spaces the dual will be denoted by $X^\ast$ and always assumed to carry the induced Banach space topology.
For locally convex spaces $X, Y$ and continuous linear operators $T : X \to Y$, we shall denote the transpose by $\ltrans{T} : Y' \to X'$.
The inner product on $L^2 ( \mathbb{R}^d )$ will be denoted with $( \cdot, \cdot )$ and the associated continuous inclusion of $\mathcal{S}$ into $\mathcal{S}'_\beta$ by $\iota : \mathcal{S} \to \mathcal{S}'_\beta$.

\begin{definition}
Let $(X, p)$ be a seminormed space and consider the set $\mathcal{C}(X)$ of all Cauchy sequences in $X$ and let $(x_n) \sim (y_n)$ whenever $\lim_{n \to \infty} p( x_n - y_n ) = 0$.
Then $\mathcal{C}(X) / \sim$ is a vector space with the obvious operations, $\bar{p} ( [(x_n)]_\sim ) = \lim_{n \to \infty} p( x_n )$ is a well-defined norm on $\mathcal{C}(X) / \sim$ and we shall denote the resulting normed space by $X_p$.
$X_p$ is complete and is called the \textbf{completion} of $(X,p)$.
There is also a \textbf{natural map} $\pi^p : X \to X_p, x \mapsto [(x, x, \dots)_\sim]$ which is linear, continuous and has dense range.
\end{definition}

\begin{remark}
Let $(X,p)$ be a seminormed space and $Y$ a complete locally convex space.
Then, by the Hahn-Banach theorem, every continuous linear operator $L : X \to Y$ extends uniquely to the completion $X_p$ in the sense that there is a unique continuous linear operator $\bar{L} : X_p \to Y$ such that $L = \bar{L} \circ \pi^p$.
Consequently, for every seminorm $q \ge p$ on $X$ there is a unique continuous, linear \textbf{natural map} $\pi^p_q : X_q \to X_p$ with $\pi^p = \pi^p_q \circ \pi^q$.
\end{remark}
\subsection{Measure Theory}
To a great part the results of this paper depend on properties of Gaußian measures on locally convex vector spaces.
As such we will give some definitions of the relevant concepts taken mostly from \cite{src:Bogachev:GaußianMeasures}.
\begin{definition}
Given a topological space $X$, we let $\mathcal{B}(X)$ denote its \textbf{Borel} $\sigma$-algebra, i.e. the smallest $\sigma$-algebra containing all open subsets of $X$.
A member of $\mathcal{B}(X)$ is called a \textbf{Borel set}.
A measure $\mu$ on $\mathcal{B}(X)$ is called a \textbf{Borel measure on $X$}.

Given a locally convex space $X$, we define the \textbf{cylindrical $\bm{\sigma}$-algebra} $\mathcal{E}(X) \subseteq \mathcal{B}(X)$ to be the smallest $\sigma$-algebra with respect to which every function in $X'$ is measurable.
\end{definition}
\begin{definition}
For any finite measure $\mu$ on $\mathcal{E}(X)$ where $X$ is a locally convex space, we define its \textbf{characteristic function} as
\begin{equation}
\hat{\mu} \left( \phi \right) = \int_X \exp \left[ i \phi \left( x \right) \right] \mathrm{d} \mu \left( x \right)
\end{equation}
for all $\phi \in X'$.
\end{definition}
\begin{definition}
For any measure $\mu$ on $\mathcal{E}(X)$ where $X$ is a locally convex space, we define its \textbf{moment-generating function} $Z : X' \to \bar{\mathbb{R}}$ as
\begin{equation}
Z \left( \phi \right) = \int_X \exp \left[ \phi \left( x \right) \right] \mathrm{d} \mu \left( x \right)
\end{equation}
for all $\phi \in X'$.
\end{definition}
The following lemma is immediate from Hölder's inequality and Fatou's lemma.
\begin{lemma}
\label{lem:MomentGeneratingFunctionIsLSCAndLogarithmicallyConvex}
The moment-generating function of a finite measure is logarithmically convex, proper convex and lower semicontinuous whenever $X'$ is equipped with a topology at least as fine as the weak\nobreakdash-$\ast$ topology.
\end{lemma}
\begin{definition}
Let $\mu$ be a measure on a $\sigma$-algebra $\mathcal{A}$ of subsets of a set $X$ and $f : X \to Y$ a function into another set $Y$ equipped with a $\sigma$-algebra $\mathcal{A}'$.
If $f^{-1}(A') \in \mathcal{A}$ whenever $A' \in \mathcal{A}'$, then $f_* \mu := \mu \circ f^{-1}$ is a measure on $\mathcal{A}'$ called the \textbf{pushforward measure of $\mu$ under $f$}.
\end{definition}
\begin{definition}
A measure $\mu$ on $\mathcal{E}(X)$ where $X$ is a locally convex space is a \textbf{centred Gaußian measure} if the pushforward measures $\phi_* \mu$ are centred Gaußian Borel measures on $\mathbb{R}$ for every $\phi \in X'$.
A Borel measure $\mu$ on $X$ is a centred Gaußian measure if its restriction to $\mathcal{E}(X)$ is.
\end{definition}
\begin{definition}
Let $X$ be a topological space.
A finite Borel measure $\mu$ is a \textbf{Radon measure} if, for every Borel set $B \subseteq X$ and every $\epsilon > 0$, there exists a compact set $K \subseteq B$ such that $\mu ( B \setminus K ) < \epsilon$.
\end{definition}
\begin{lemma}[{\cite[Appendix 3]{src:Bogachev:GaußianMeasures}}]
A finite Radon measure on a locally convex space is uniquely determined by its characteristic function.
\end{lemma}
\begin{remark}
It is well known that Radon measures on locally convex spaces are supported on compactly embedded separable reflexive Banach spaces.
We shall however not make use of this fact and always integrate over the whole locally convex space in question i.e. including any null sets outside such supports of the relevant measures.
\end{remark}
\begin{definition}
\label{def:CameronMartinSpace}
Let $\mu$ be a Radon Gaußian measure on a locally convex space $X$.
Note that $X' \subset L^2(\mu)$ by definition and denote by $X'_\mu$ the closure of $X'$ in $L^2(\mu)$.
Given $f \in X'_\mu$, there exists a unique $R_\mu f \in X$ such that \cite[Theorem 3.2.3]{src:Bogachev:GaußianMeasures}
\begin{equation}
\phi \left( R_\mu f \right) = \int_X \phi f \mathrm{d} \mu \qquad \mbox{for all} \ \phi \in X' \, .
\end{equation}
We now define the \textbf{Cameron-Martin space} of $\mu$ as the range $H(\mu) = R_\mu(X'_\mu)$, which is turned into a separable Hilbert space by the inner product induced by $L^2(\mu)$ \cite[Theorem 3.2.7]{src:Bogachev:GaußianMeasures}.
Then, $R_\mu : X_\mu' \to H(\mu)$ is a Hilbert space isomorphism and and the canonical inclusion $H(\mu) \to X$ is continuous.
\end{definition}
\begin{theorem}[{\cite[Theorem 2.2.4]{src:Bogachev:GaußianMeasures}}]
Let $\mu$ be a centred Radon Gaußian measure on a locally convex space $X$.
Then, its characteristic function takes the form
\begin{equation}
\hat{\mu} \left( \phi \right) = \exp \left[ - \frac{1}{2} \phi \left( R_\mu \phi \right) \right] = \exp \left[ - \frac{1}{2} \left \Vert \phi \right \Vert_{L^2(\mu)}^2 \right]
\end{equation}
for every $\phi \in X'$.
\end{theorem}
\begin{theorem}[Cameron-Martin {\cite[Corollary 2.4.3, Remark 3.1.8]{src:Bogachev:GaußianMeasures}}]
\label{thm:CameronMartinTheorem}
Let $\mu$ be a centred Radon Gaußian measure on a locally convex space $X$ and $h \in H(\mu)$ an element of its Cameron-Martin space.
Then the pushforward measure $\mu_h = \mu \circ \tau_h^{-1}$ with $\tau_h : X \to X, x \mapsto x - h$ is equivalent to $\mu$ with the corresponding Radon-Nikodym derivative given by
\begin{equation}
\frac{\mathrm d \mu_h}{\mathrm{d} \mu} \left( x \right)
=
\exp \left[ \left( R_\mu^{-1} h \right) \left( x \right) - \left \Vert h \right \Vert_{H(\mu)}^2 \right]
\end{equation}
for all $x \in X$.
\end{theorem}
\begin{lemma}
\label{lem:BallHasNonzeroMeasure}
Let $\mu$ be a centred Radon Gaußian probability measure on a locally convex space $X$ and $p$ a continuous seminorm on $X$.
Furthermore, let $B^p_r(x)$ denote the open $p$-ball of radius $r$ around $x \in X$.
Then, for every $\epsilon > 0$
\begin{equation}
\mu \left( B^p_{\epsilon} \left(0 \right) \right) > 0 \, .
\end{equation}
\end{lemma}
\begin{proof}
Suppose there is some $\epsilon > 0$ such that $p^{-1} ( [ 0, \epsilon ) )$ has zero $\mu$-measure.
By \cite[theorem 3.6.1]{src:Bogachev:GaußianMeasures}, the closure $\overline{H(\mu)}$ in $X$ has full $\mu$-measure.
Since $H(\mu)$ is separable, it has some countable dense set $S$ and since the inclusion $H(\mu) \to X$ is continuous we arrive at the contradiction
\begin{equation}
1
=
\mu \left( \overline{H(\mu)} \right)
\le
\mu \left( \bigcup_{h \in S} B^p_{\epsilon} \left( h \right) \right)
\le
\sum_{h \in S} \mu \left( B^p_{\epsilon} \left( h \right) \right)
=
0 \, ,
\end{equation}
where the last equality follows from theorem \ref{thm:CameronMartinTheorem}.
\end{proof}
\begin{definition}
\label{def:MeasureConvergence}
A sequence $(\omega_n)_{n \in \mathbb{N}}$ of Radon measures on a locally convex space $X$ \textbf{converges weakly} to another Radon measure $\omega$ if
\begin{equation}
\lim_{n \to \infty} \int_X f \, \mathrm{d} \omega_n
=
\int_X f \, \mathrm{d} \omega
\end{equation}
for all bounded continuous functions $f : X \to \mathbb{R}$.
\end{definition}
\begin{theorem}[Portmanteau theorem {\cite[Theorem 3.8.2]{src:Bogachev:GaußianMeasures}}]
A sequence of Radon probability measures $(\mu_n)_{n \in \mathbb{N}}$ on a locally convex space $X$ converges weakly to a Radon probability measure $\mu$ on $X$ precisely when either (and then both) of the following conditions is satisfied:
\begin{itemize}
\item $\liminf_{n \to \infty} \mu_n (U) \ge \mu(U)$ for every open set $U \subseteq X$,
\item $\liminf_{n \to \infty} \mu_n (C) \le \mu(C)$ for every closed set $C \subseteq X$.
\end{itemize}
\end{theorem}
\begin{lemma}
\label{lem:IntegrationIsLowerSemicontinuousWRTWeakConvergence}
Let $X$ be a locally convex space, $(\mu_n)_{n \in \mathbb{N}}$ a sequence of Borel probability measures on X weakly converging to a Radon measure $\mu$ on $X$ and $f : X \to \mathbb{R}$ a lower semicontinuous function that is bounded from below.
Then
\begin{equation}
\int_X f \mathrm{d} \mu \le \liminf_{n \to \infty} \int_X f \mathrm{d} \mu_n \, .
\end{equation}
\end{lemma}
\begin{proof}
From \cite[Corollary 8.2.5]{src:Bogachev:MeasureTheory2} this is true if $f$ is bounded.
For unbounded $f$, set $f_m = \max \{f, m \}$ which is bounded and lower semicontinuous.
Then,
\begin{equation}
\int_X f \mathrm{d} \mu
=
\sup_{m \in \mathbb{N}} \int_X f_m \mathrm{d} \mu
\le
\sup_{m \in \mathbb{N}} \liminf_{n \to \infty} \int_X f_m \mathrm{d} \mu_n
\le
\liminf_{n \to \infty} \int_X f \mathrm{d} \mu_n \, .
\end{equation}
\end{proof}
\begin{definition}
A sequence $(\mu_n)_{n \in \mathbb{N}}$ of finite Borel measures on a topological space $X$ is \textbf{uniformly tight} if for any $\epsilon > 0$ there exists a compact set $K \subseteq X$ such that $\mu_n(X \setminus K) < \epsilon$ for all $n \in \mathbb{N}$.
\end{definition}

\subsection{Convex Functions}
A convex function is \textbf{proper} if it does not attain the value $-\infty$ and is not equal to the constant function $\infty$.
\begin{definition}
Let $X$ be a Hausdorff, locally convex topological vector space and $f : X \to \bar{\mathbb{R}}$ a proper convex and lower semicontinuous function.
Then, the \textbf{convex conjugate (Legendre-Fenchel transform)} $f^c : X' \to \bar{\mathbb{R}}$ of $f$ is defined as
\begin{equation}
\phi \mapsto \sup_{T \in X} \left[ \phi \left( T \right) - f \left( T \right) \right]
\end{equation}
for all $\phi \in X'$.
If we equip $X'$ with a topology $\tau$ at least as fine as the weak\nobreakdash-$\ast$ topology. it is also proper convex and lower semicontinuous.
We may then also define $(f^c)^c : (X', \tau)' \to \bar{\mathbb{R}}$ and by the well-known \textbf{Fenchel-Moreau theorem} $(f^c)^c \vert_{X} = f$ \cite{src:Zalinescu:ConvexAnalysisInGeneralVectorSpaces}.
\end{definition}
\begin{theorem}[{\cite[Theorem 2.2.9]{src:Zalinescu:ConvexAnalysisInGeneralVectorSpaces}}]
\label{thm:ConvexAndSomewhereBoundedFromAboveImpliesContinuous}
Let $f : X \to \mathbb{R}$ be a convex function on a Hausdorff, locally convex space $X$.
If $f$ is bounded from above on some open subset of $X$, then $f$ is continuous.
\end{theorem}
\begin{theorem}[{\cite[Theorem 2.2.20]{src:Zalinescu:ConvexAnalysisInGeneralVectorSpaces}}]
\label{thm:ConvexAndLowerSemicontinuousImpliesContinuous}
Let $f : X \to \mathbb{R}$ be a convex and lower semicontinuous function on $X$ where $X$ is a Banach space or a reflexive space.
Then $f$ is continuous.
\end{theorem}
\begin{lemma}
\label{lem:ConvexLowerSemicontinuityConvergenceImpliesContinuous}
Let $X$ be a Fréchet space.
Furthermore, let $f_n : X \to \bar{\mathbb{R}}$ be a sequence of convex and lower semicontinuous functions converging pointwise to a function $f : X \to \mathbb{R}$.
Then $f$ is continuous.
\end{lemma}
\begin{proof}
By the pointwise convergence, $f$ is clearly convex.
Hence, by theorem \ref{thm:ConvexAndSomewhereBoundedFromAboveImpliesContinuous} it suffices to show that $Z$ is bounded from above on some open subset of $X$.
Let
\begin{equation}
A_{K, N} = \bigcap_{n \in \mathbb{N}_{\ge N}} f_n^{-1} \left( \left( - \infty, K \right] \right)
\end{equation}
for $K, N \in \mathbb{N}$.
Then all $A_{K,N}$ are closed by the lower semicontinuity of $f_n$.
Furthermore, $\bigcup_{K, N \in \mathbb{N}} A_{K, N} = X$ because $\lim_{n \to \infty} f_n (x) = f(x) < \infty$ for all $x \in X$.
By the Baire category theorem, some $A_{K,N}$ contains an open set, i.e. there exists $N, K \in \mathbb{N}$, $x \in X$ and an open neighbourhood $U \subseteq X$ of zero such that
\begin{equation}
\sup_{y \in U} f_n \left( x + y \right) \le K
\end{equation}
for all $n \in \mathbb{N}_{\ge N}$.
Thus $f$ is bounded from above on $x + U$.
\end{proof}
\section{A Theorem on Lower Semicontinuous Envelopes}
\label{sec:LSCEnvelopes}
In the context of the Wetterich equation, a major role is played by the effective average action.
It is formally defined via the Legendre-Fenchel transform of the logarithm of the partition function which is convex and lower semicontinuous by lemma \ref{lem:MomentGeneratingFunctionIsLSCAndLogarithmicallyConvex}.
Hence, it is clear that the study of the effective average action necessitates some features of convex analysis.

In this section we introduce the concepts of supercoercivity and lower semicontinuous envelopes and prove some simple results that the author believes to be novel.

\begin{definition}
Let $X$ be a locally convex space and $f : X \to \bar{\mathbb{R}}$ a convex function with
\begin{equation}
\sup_{x \in X} \left[ p \left( x \right) - f \left( x \right) \right] < \infty
\end{equation}
for all continuous seminorms $p$ on $X$.
Then $f$ is \textbf{supercoercive}.
\end{definition}
\begin{definition}
Let $X$ be a locally convex space, $Y$ a normed space, $\iota: X \to Y$ linear and continuous withe dense range and $f : X \to \bar{\mathbb{R}}$ a convex and lower semicontinuous function.
Then the \textbf{lower semicontinuous envelope} $LSC(f,\iota) : Y \to \bar{\mathbb{R}}$ of $f$ with respect to $\iota$ is given by
\begin{equation}
LSC(f,\iota) \left( x \right) = \inf \left \{ \liminf_{n \to \infty} f \left( x_n \right) \, \big \vert \, (x_n)_{n \in \mathbb{N}} \text{ in $X$ with } \lim_{n \to \infty} \left \Vert \iota \left( x_n \right) - x \right \Vert_Y = 0 \right \}
\end{equation}
\end{definition}
\begin{remark}
This definition is a generalisation of the one that is commonly used in literature, see e.g. \cite{src:Zalinescu:ConvexAnalysisInGeneralVectorSpaces}.
It is easy to see that $LSC(f,\iota)$ is lower semicontinuous.
Furthermore, if $\iota$ is injective, $LSC(f,\iota)$ is the largest lower semicontinuous function that is not greater than $f \circ \iota^{-1}$ on $\iota(X)$.
\end{remark}
The following lemma shows that our definition of lower semicontinuous envelopes plays well with Legendre-Fenchel conjugation.
\begin{lemma}
\label{lem:LSCConjugation}
Let $X$ be a Hausdorff reflexive space, $p$ a continuous seminorm on $X'_\beta$, $\iota_p : X'_\beta \to (X'_\beta)_p = Y$ the natural map and $f : X \to \bar{\mathbb{R}}$ a proper convex and lower semicontinuous function.
Let $f^p$ denote the restriction of $f$ to the Banach space $Y^*$ considered as a subspace of $X$.
Then $LSC(f^c,\iota_p)^c = f^p$ and if $f^p$ is proper, $LSC(f^c,\iota_p) = (f^p)^c \vert_Y$.
\end{lemma}
\begin{remark}
The transpose $\ltrans{\iota}_p : Y^* \to X$ is injective by \cite[Chapter 4, §4, Corollary 2.3]{src:SchaeferWolff:TopologicalVectorSpaces}.
\end{remark}
\begin{proof}
By Fenchel-Moreau it suffices to show that $LSC(f^c,\iota_p)^c = f^p$.
Plugging in the definition of $LSC(f^c,\iota_p)^c$, for every $x \in Y^*$ there is a sequence $(\phi_n)_{n \in \mathbb{N}}$ in $Y$ and a sequence $(\psi_n)_{n \in \mathbb{N}}$ in $X'_\beta$ with $\lim_{n \to \infty} p(\phi_n - \iota_p \psi_n) = 0$ such that
\begin{equation}
LSC(f^c,\iota_p)^c \left( x \right)
=
\lim_{n \to \infty} \left[ x \left( \phi_n \right) - f^c \left( \psi_n \right) \right] \, .
\end{equation}
But since $x \in Y^*$, we have
\begin{equation}
\lim_{n \to \infty} \left \vert x \left( \phi_n - \iota_p \psi_n \right) \right \vert
\le
\lim_{n \to \infty} C p \left( \phi_n - \iota_p \psi_n \right)
=
0
\end{equation}
for some $C > 0$.
Consequently,
\begin{equation}
LSC(f^c,\iota_p)^c \left( x \right)
=
\lim_{n \to \infty} \left[ x \left( \iota_p \psi_n \right) - f^c \left( \psi_n \right) \right]
\le
f^{cc} \left( \ltrans{\iota}_p x \right)
=
f^p \left( x \right) \, .
\end{equation}
For the converse inequality, note that
\begin{align}
LSC \left(f^c,\iota_p\right)^c \left( x \right)
&\ge
\sup_{\phi \in X'_\beta} \left[ x \left( \iota_p \phi \right) - LSC(f^c,\iota_p) \left( \iota_p \phi \right) \right] \\
&\ge
\sup_{\phi \in X'_\beta} \left[ x \left( \iota_p \phi \right) - f^c \left( \phi \right) \right]
=
f^{cc} \left(  \ltrans{\iota}_p x \right)
=
f^p \left( x \right) \, . \qedhere
\end{align}
\end{proof}
While this lemma demonstrates a useful property, it may in general be difficult to actually calculate the lower semicontinuous envelope of a given function.
In the supercoercive case, however, we obtain a particularly simple expression.
\begin{lemma}
\label{lem:LSCEnvelopeSupercoercive}
Let $X$ be a Hausdorff reflexive space and $f : X \to \bar{\mathbb{R}}$ a convex, lower semicontinuous and supercoercive function.
For any continuous seminorm $p$, let $\iota_p : X \to X_p$ denote the natural map.
Then the lower semicontinuous envelope $g$ of $f$ with respect to $\iota_p$ takes the form
\begin{equation}
g \left( x \right) = \begin{cases}
\inf \left \{ f \left( y \right) : y \in \iota_p^{-1} \left( \left \{ x \right \} \right) \right \}& x \in \iota_p \, , \left(X\right) \\
\infty & \text{otherwise}
\end{cases}
\end{equation}
and is supercoercive.
\end{lemma}
\begin{proof}
It is immediately clear that $LSC(f,\iota_p) (x) \le g(x)$ for all $x \in X_p$.
Suppose that $g(x) > LSC(f,\iota_p) (x)$ for some $x \in X_p$.
Then there is a sequence $(x_n)_{n \in \mathbb{N}}$ in $X$ with $\lim_{n \to \infty} p(x - \iota_p x_n) = 0$ such that $\lim_{n \to \infty} f (x_n) < g(x)$.
If $x_n$ is bounded in $X$ there is a subnet $(y_\alpha)$ that is weakly converging to some $y \in X$ with $x = \iota_p (y)$.
But $f$ is weakly lower semicontinuous by \cite[Theorem 2.2.1]{src:Zalinescu:ConvexAnalysisInGeneralVectorSpaces} such that
\begin{equation}
f \left( y \right) \le \lim_{\alpha} f \left( y_\alpha \right) = \lim_{n \to \infty} f \left( x_n \right) < g \left( x \right) \, ,
\end{equation}
which contradicts the definition of $g$.

If $x_n$ is unbounded in $X$ there is some subsequence $(y_n)_{n \in \mathbb{N}}$ and a continuous seminorm $q$ on $X$ with $\lim_{n \to \infty} q (y_n) = \infty$.
By the supercoercivity of $f$, we must then have $\lim_{n \to \infty} f(x_n) = \lim_{n \to \infty} f(y_n) = \infty$ which is again a contradiction.

For the supercoercivity, let $(x_n)_{n \in \mathbb{N}}$ be any sequence in $\iota_p(X)$.
Then for every $x_n$ there is also some $y_n \in \iota_p^{-1}(\{x_n\})$ with
\begin{equation}
\left \vert f \left( y_n \right) - g \left( x_n \right) \right \vert < \frac{1}{n} \, .
\end{equation}
By the definition of $\iota_p$, $p(x_n) = p(y_n)$ such that
\begin{equation}
\limsup_{n \to \infty} \left[ p \left( x_n \right) - g \left( x_n \right) \right]
\le
\limsup_{n \to \infty} \left[ p \left( y_n \right) - f \left( y_n \right) + \frac{1}{n} \right]
< \infty
\end{equation}
by the supercoercivity of $f$.
Because the sequence $(x_n)$ was arbitrary, $g$ is supercoercive.
\end{proof}
\section{A Renormalisation Theorem}
\label{sec:Convergence}
In 1973 Osterwalder and Schrader gave a rigorous interpretation of Quantum Field Theory as axiomatised by Gårding and Wightman \cite{src:GaardingWightman} in terms of Wick-rotated correlation functions \cite{src:OsterwalderSchrader1,src:OsterwalderSchrader2}.
A modern formulation in terms of a measure on a space of distributions may be found in the book by Glimm and Jaffe \cite{src:GlimmJaffe} and a comprehensive introduction to the Gårding-Wightman axioms is given in the work by Jost \cite{src:Jost}.

The task of constructing measures reflecting physical aspects of reality typically requires a parameter-dependent regularisation.
The parameters can be thought of as a directed set of \enquote{windows of applicability} and the corresponding measures are thought to reflect physics more or less accurately within such a window.
Common choices are the replacement of $\mathbb{R}^d$ with a finite lattice or the introduction of certain \enquote{cutoffs} as in perturbative Quantum Field Theory, leading to models that agree well with experiments e.g within specific ranges of characteristic momenta $[ p_{\mathrm{min}}, p_{\mathrm{max}}]$.
The validity of this methodology comes from experiments in remarkable agreement with concrete calculations \cite{src:QEDPrecisionExperiment}.

Collecting such regularised models in a box results in a net $( \omega_\alpha )_{\alpha \in I}$ of measures.
The big question that remains is whether it is possible to remove the regularisation completely, i.e to interpret the limit $\lim_\alpha \omega_\alpha$ as a bona fide measure satisfying the Osterwalder-Schrader axioms.
Perhaps the most notable example where this program has worked out is the $\phi^4_3$ model \cite{src:phi43}.

In this work, the relevant measures will live on $\mathcal{S}'_\beta$ and we apply the mode of convergence given in definition \ref{def:MeasureConvergence}.

\begin{theorem}[Lévy continuity theorem \cite{src:Meyer:LevyContinuityTheorem}]
\label{thm:LevyContinuity}
Let $(\omega_n)_{n \in \mathbb{N}}$ be a sequence of Borel measures on $\mathcal{S}'_\beta$ such that their characteristic functions converge pointwise to a function that is continuous at zero.
Then there is a Radon measure $\omega$ on $\mathcal{S}'_\beta$ such that $\omega_n$ converges weakly to $\mu$.
\end{theorem}

However, in the context of Wetterich's equation the objects under consideration are not the characteristic functions but rather the convex conjugates of the logarithms of corresponding moment-generating functions $Z_\alpha : \mathcal{S} \to \bar{\mathbb{R}}$ with
\begin{equation}
\phi \mapsto \int_{\mathcal{S}'_\beta} \exp \left[ T \left( \phi \right) \right] \mathrm{d} \omega_\alpha \left( T \right) \, .
\end{equation}

Consequently, we shall work out a convergence theorem that deals with (the logarithms of) moment-generating functions and their convex conjugates.

The convergence of convex conjugate functions in infinite dimensions was originally studied on reflexive Banach spaces by Mosco who introduced the now-called Mosco convergence \cite{src:Mosco:MoscoConvegence}.
It was later generalised to locally convex spaces by Beer and Borwein \cite{src:BeerBorwein}, de Acosta \cite{src:deAcosta} and Zabell \cite{src:Zabell:MoscoLocallyConvex}.
Following the latter, we give the following definition:\footnote{In \cite{src:Zabell:MoscoLocallyConvex} Zabell gives the definition of Mosco convergence in terms of Mackey convergence. On Banach spaces and reflexive locally convex spaces, the notions of norm convergence and Mackey convergence coincide \cite[Chapter 4, Theorem 3.4]{src:SchaeferWolff:TopologicalVectorSpaces} and since we shall not work on more general spaces, the given definitions suffice.}
\begin{definition}
Let $X$ be a Banach space or a reflexive, locally convex space and $(f_n)_{n = 0}^{\infty} : X \to \bar{\mathbb{R}}$ be a sequence of extended real-valued functions on $X$.
Then
\begin{itemize}
\item $f_n$ \textbf{(M1)-converges} to $f_0$ if for every $x \in X$ there is some sequence $(x_n)_{n \in \mathbb{N}}$ that converges to $x$ such that
\[ \limsup_{n \to \infty} f_n \left( x_n \right) \le f_0 \left( x \right) \, .\]
\item $f_n$ \textbf{(K2)-converges} to $f_0$ if for every $x \in X$ and every sequence $(x_n)_{n \in \mathbb{N}}$ that converges to $x$
\[ \liminf_{n \to \infty} f_n \left( x_n \right) \ge f_0 \left( x \right) \, .\]
\item $f_n$ \textbf{(M2)-converges} to $f_0$ if for every $x \in X$ and every sequence $(x_n)_{n \in \mathbb{N}}$ that converges weakly to $x$
\[ \liminf_{n \to \infty} f_n \left( x_n \right) \ge f_0 \left( x \right) \, .\]
\end{itemize}
If $f_n$ (M1)- and (K2)-converges to $f_0$, we shall say that $f_n$ \textbf{epi-converges} to $f_0$ or converges to $f_0$ in the \textbf{Painlevé-Kuratowski} sense \cite[Theorem 5.3.5]{src:Beer:TopologiesOnClosedAndClosedConvexSets}.
If $f_n$ (M1)- and (M2)-converges to $f_0$, we shall say that $f_n$ \textbf{Mosco converges} to $f_0$.
\end{definition}
\begin{remark}
Clearly, Mosco convergence implies epi-convergence.
\end{remark}
Using Mosco convergence it is possible to express a continuity theorem of the Legendre-Fenchel conjugation.
As Zabell proved in \cite{src:Zabell:MoscoLocallyConvex}, however, the Legendre-Fenchel conjugation is not a homeomorphism with respect to Mosco convergence.
A stronger notion offering this feature is given by the so-called \textbf{Attouch-Wets convergence} which we shall exploit in our final theorem.
Its precise formulation is somewhat complicated such that we refer to \cite{src:EpigraphicalAndUniformConvergence} for a definition.
In fact, we will not need to prove Attouch-Wets convergence from first principles, but only indirectly through the theorems presented in \cite{src:EpigraphicalAndUniformConvergence} such that a lack of definition appears tolerable to the author.
Another great achievement of the Attouch-Wets convergence, is its compatibility with pointwise convergence which was also worked out in \cite{src:EpigraphicalAndUniformConvergence}.

Let us begin with the following simple lemma which trivially follows from the metrisability of $\mathcal{S}$.
\begin{lemma}
\label{lem:NullSequenceDivergence}
Let $( \phi_n )_{n \in \mathbb{N}}$ be a null sequence in $\mathcal{S}$.
Then, there exists a monotonically increasing sequence $( t_n )_{n \in \mathbb{N}}$ in $\mathbb{R}_{> 0}$ such that
\begin{equation}
\lim_{n \to \infty} t_n = \infty
\qquad \text{and} \qquad
\lim_{n \to \infty} t_n \phi_n = 0 \, .
\end{equation}
\end{lemma}
For the remainder of the section, we will use the following abbreviation.
\begin{definition}
For a continuous seminorm $p$ on $\mathcal{S}'_\beta$, let $\iota_p : \mathcal{S}'_\beta \to \mathcal{S}'_p$ denote the natural map to the corresponding completion.
Furthermore, for any function $f$ on $\mathcal{S}$ let $f^p$ denote its restriction to $( \mathcal{S}'_p )^*$.
\end{definition}
\begin{remark}
$f^p$ is well-defined because $\iota_p$ has dense range, whence its transpose $\ltrans{\iota}_p : (\mathcal{S}'_p)^* \to \mathcal{S}$ is injective \cite[Chapter 4, Corollary 2.3]{src:SchaeferWolff:TopologicalVectorSpaces}.
\end{remark}
We may now formulate a sufficient condition for the weak convergence of a sequence of measures.
\begin{theorem}
\label{thm:ConvergenceEquivalentMosco}
Let $( \omega_n )_{n \in \mathbb{N}}$ be a sequence of Borel probability measures on $\mathcal{S}'_\beta$ and $( Z_n )_{n \in \mathbb{N}}$ the corresponding moment-generating functions.
Suppose that
\begin{equation}
\limsup_{n \to \infty} Z_{n} \left( \phi \right) < \infty
\end{equation}
for all $\phi \in \mathcal{S}$.
Then, $\omega_n$ converges weakly to another Borel probability measure $\omega$ if and only if
\begin{itemize}
\item $Z_n$ Mosco converges to a convex, lower semicontinuous function $Z : \mathcal{S} \to \mathbb{R}$ and
\item for all continuous seminorms $p$ on $\mathcal{S}'_\beta$ the restrictions $Z_n^p$ Mosco converge to the corresponding restriction $Z^p$.
\end{itemize}
Moreover, in the affirmative case $Z$ is continuous and the moment-generating function of $\omega$.
\end{theorem}
\begin{proof}
Suppose $\omega_n \to \omega$ weakly and fix some $\phi \in \mathcal{S}$.
Then, we clearly have pointwise convergence of the characteristic functions $\widehat{\phi_\ast \omega_n}$ of the one-dimensional pushforward measures to $\widehat{\phi_\ast \omega}$.
Because $Z_n(\phi)$ and $Z_n(-\phi)$ are eventually finite it is known that $\widehat{\phi_\ast \omega_n}$ has an analytic continuation  $\overline{\phi_\ast \omega_n}$ to $\{ z \in \mathbb{C} : \vert z \vert \le 1 \}$ satisfying \cite{src:Lukacs:CharacteristicFunctions}
\begin{equation}
\label{eq:LukacsBound}
\max \left \{ Z_n \left( \phi \right), Z_n \left( -\phi \right) \right \}
=
\sup_{\left \vert z \right \vert \le 1} \overline{\phi_\ast \omega_n}
\left( z \phi \right)
\end{equation}
where $z$ is now a complex variable.
Hence the family $\{ \overline{\phi_\ast \omega_n} : n \in \mathbb{N} \}$ is eventually uniformly bounded within the unit ball of $\mathbb{C}$.
By the Vitali-Porter theorem \cite{src:Schiff:NormalFamilies}, pointwise convergence on the real axis implies pointwise convergence on the imaginary axis.
Since this is true for all $\phi \in \mathcal{S}$, we obtain pointwise convergence of $Z_n$ to some real-valued function $Z$ which is convex and continuous by lemma \ref{lem:ConvexLowerSemicontinuityConvergenceImpliesContinuous}.
This clearly implies that $Z_n$ (M1)-converges to $Z$.
Moreover, by lemma \ref{lem:IntegrationIsLowerSemicontinuousWRTWeakConvergence}, the moment-generating function $M$ of $\omega$ is bounded by $Z$ and is, in particular, finite everywhere.
But then the analytic continuations of $t \to M(t \phi)$ and $t \to Z(t \phi)$ to the imaginary axis must both equal the characteristic function $\widehat{\phi_\ast \omega}$.
Consequently, $M = Z$.
Regarding the (M2)-convergence, it is straightforward to see that the methods used for proving \cite[theorem~1.1]{src:VaryingFatou} apply also in our case such that for every weakly convergent sequence $\phi_n \to \phi$ in $\mathcal{S}$, we have
\begin{equation}
\begin{aligned}
Z(\phi)
&=
\int_{\mathcal{S}'_\beta} \exp \left[ \phi \left(T \right) \right] \mathrm{d} \omega \left( T \right)
=
\int_{\mathcal{S}'_\beta} \liminf_{n \to \infty, T' \to T} \exp \left[ \phi_n \left(T' \right) \right] \mathrm{d} \omega \left( T \right) \\
&\le
\liminf_{n \to \infty} \int_{\mathcal{S}'_\beta} \exp \left[ \phi_n \left(T \right) \right] \mathrm{d} \omega_n \left( T \right)
= \liminf_{n \to \infty} Z_n \left( \phi_n \right) \, ,
\end{aligned}
\end{equation}
where $T' \to T$ is considered in the topology of $\mathcal{S}'_\beta$.
The second equality follows from the boundedness of $(\phi_n)_{n \in \mathbb{N}}$ by the definition of the strong dual topology.
Hence, $Z_n$ (M2)-converges to $Z$.
Now, let $p$ be a continuous seminorm on $\mathcal{S}'_\beta$.
Clearly $Z^p_n$ (M2)-converges to $Z^p$ since every weakly convergent sequence in $( \mathcal{S}'_p )^*$ is also weakly convergent in $\mathcal{S}(\mathbb{R}^d)$.
Finally, the pointwise convergence ensures (M1)-convergence of $Z_n^p$ to $Z^p$.

Conversely, assume that $Z_n$ Mosco converges to $Z$ and that the same holds for the corresponding restrictions as above.
For any $\phi \in \mathcal{S}$ fix a continuous seminorm $p$ on $\mathcal{S}'_\beta$ such that $\phi \in ( \mathcal{S}'_p )^*$.
Then, by assumption, $Z^p_n$ Mosco converges to $Z^p$ and, in particular, also converges in the Painlevé-Kuratowski sense.
Furthermore, because all $Z_n$ are lower semicontinuous (see lemma \ref{lem:MomentGeneratingFunctionIsLSCAndLogarithmicallyConvex}), so are all $Z_n^p$.
Since $\limsup_{n \to \infty} Z^p_n ( \psi ) < \infty$ for all $\psi \in ( \mathcal{S}'_p )^*$, we have that $Z^p_n$ converges pointwise to $Z^p$ by \cite[Corollary 2.3.]{src:EpigraphicalAndUniformConvergence}.
Since, $\phi \in \mathcal{S}$ was arbitrary, we obtain $Z_n \to Z$ pointwise and the continuity of $Z$ by lemma \ref{lem:ConvexLowerSemicontinuityConvergenceImpliesContinuous}.

In analogy to the first part of the proof, we may now use the pointwise convergence of $\overline{\phi_\ast \omega_n}$ along the imaginary axis, the bound given in equation \ref{eq:LukacsBound} and the Vitali-Porter theorem to conclude that $\hat{\omega}_n$ converges pointwise to some function $c : \mathcal{S}(\mathbb{R}^d) \to \mathbb{C}$.
By theorem \ref{thm:LevyContinuity}, it just remains to show that $c$ is continuous at zero.
By the continuity of $Z$, there is some balanced neighbourhood $U \subseteq S$ of the origin such that
\begin{equation}
\sup_{\phi \in U} \left \vert Z \left( \phi \right) - Z \left( 0 \right) \right \vert \le 1
\qquad\text{and hence}\qquad
0 \le \sup_{\phi \in U} Z \left( \phi \right) \le 2 \, .
\end{equation}
But then, for all $t \in ( 0, 1 )$,
\begin{equation}
\begin{aligned}
\sup_{\phi \in t U} \left \vert c \left( \phi \right) - c \left( 0 \right) \right \vert
&=
\sup_{\phi \in U} \lim_{n \to \infty} \left \vert \hat{\omega}_n \left( t \phi \right) - \hat{\omega}_n \left( 0 \right) \right \vert
\le
\sup_{\phi \in U} \limsup_{n \to \infty} \int_{\mathcal{S}'_\beta} t \left \vert T \left( \phi \right) \right \vert \mathrm{d} \omega_n \left( T \right) \\
&\le
\frac{t}{2} \sup_{\phi \in U} \limsup_{n \to \infty} \left[ Z_n \left( \phi \right)  + Z_n \left(- \phi \right) \right]
\le
2 t \, .
\end{aligned}
\end{equation}
Since every null net $(\phi_\alpha)_{\alpha \in I}$ in $\mathcal{S}$ is eventually in $t U$ for all $t \in ( 0, 1 )$, the continuity of $c$ at the origin follows.
\end{proof}
A simple corollary to this theorem, is that we may in fact get rid of the Mosco convergence in $\mathcal{S}$.
\begin{corollary}
\label{cor:WeakConvergenceBanach}
Let $( \omega_n )_{n \in \mathbb{N}}$ be a sequence of Borel probability measures on $\mathcal{S}'_\beta$ and $( Z_n )_{n \in \mathbb{N}}$ the corresponding moment-generating functions.
Suppose that
\begin{equation}
\limsup_{n \to \infty} Z_{n} \left( \phi \right) < \infty
\end{equation}
for all $\phi \in \mathcal{S}$.
Then, $\omega_n$ converges weakly to another Borel probability measure $\omega$ if and only if there exists a lower semicontinuous, convex function $Z : \mathcal{S}(\mathbb{R}^d) \to \mathbb{R}$ such that for all continuous norms $p$ on $\mathcal{S}'_\beta$, the restrictions $Z_n^p$ Mosco converge to $Z^p$.

Moreover, in the affirmative case $Z$ is continuous and the moment-generating function of $\omega$.
\end{corollary}
\begin{proof}
By theorem \ref{thm:ConvergenceEquivalentMosco}, we just need to prove that the Mosco convergence of $Z_n^p$ implies that of $Z_n$.
Since for every $\phi$ there is a continuous seminorm $p$ on $\mathcal{S}'_\beta$ such that $\phi \in ( \mathcal{S}'_p )^*$, we clearly have that $Z_n$ (M1)-converges to $Z$.

For the (M2)-convergence, note that for any weakly convergent sequence $\phi_n \to \phi$ in $\mathcal{S}$, the set
\begin{equation}
B = \left \{ \phi_n : n \in \mathbb{N} \right \} \cup \left \{ \phi \right \} \subset \mathcal{S}
\end{equation}
is bounded and induces a continuous seminorm $q$ on $\mathcal{S}'_\beta$.
Since, $\mathcal{S}'_\beta$ is nuclear, $q$ is majorised by some continuous Hilbert norm $p$.
Hence, $B$ is a bounded subset of the reflexive Banach space $( \mathcal{S}'_p )^*$ which is separable because $\mathcal{S}'_p$ is.
Consequently, every subsequence of $( \phi_n )_{n \in \mathbb{N}}$ has a weakly convergent subsequence in $( \mathcal{S}'_p )^*$.
By assumption, it follows that $\phi_n$ converges weakly to $\phi$ in $( \mathcal{S}'_p )^*$ and the result follows from the (M2)-convergence of $Z_n^p$.
\end{proof}
The above corollary may appear rather inelegant and in fact we can do a lot better by using the following lemma.
\begin{lemma}
\label{lem:BanachCompactConvImpliesBoundedConv}
Let $( Z_n )_{n \in \mathbb{N}}$ be a sequence of proper convex and lower semicontinuous functions from $\mathcal{S}$ to $\bar{\mathbb{R}}$.
Given another function $Z : \mathcal{S} \to \bar{\mathbb{R}}$, the following are equivalent:
\begin{enumerate}[label=(\roman*)]
\item For all continuous seminorms $p$ on $\mathcal{S}'_\beta$, the restrictions $Z_n^p$ converge uniformly on compact sets to $Z^p$.
\item For all continuous seminorms $p$ on $\mathcal{S}'_\beta$, the restrictions $Z_n^p$ converge uniformly on bounded sets to $Z^p$.
\end{enumerate}
\end{lemma}
\begin{proof}
Let $p$ be a continuous seminorm on $\mathcal{S}'_\beta$.
By nuclearity, there is a continuous seminorm $q > p$ such that the natural map $\mathcal{S}'_q \to \mathcal{S}'_p$ has dense range and is nuclear, thus in particular is compact.
But then its transpose is compact and injective such that every bounded subset of $( \mathcal{S}'_p )^*$ is mapped injectively to a compact subset of $( \mathcal{S}'_q )^*$ on which we have uniform convergence.

The converse is clear, since every compact set is bounded.
\end{proof}
The uniform convergence on bounded sets enables the following corollary.
\begin{corollary}
\label{cor:WeakConvergenceUniform}
Let $( \omega_n )_{n \in \mathbb{N}}$ be a sequence of Borel probability measures on $\mathcal{S}'_\beta$ and $( Z_n )_{n \in \mathbb{N}}$ the corresponding moment-generating functions.
Then the following are equivalent:
\begin{enumerate}[label=(\roman*)]
\item $\limsup_{n \to \infty} Z_n ( \phi ) < \infty$ for all $\phi \in \mathcal{S}$ and $\omega_n$ converges weakly to another Borel probability measure $\omega$.
\item There exists a convex and continuous function $Z : \mathcal{S} \to \mathbb{R}$ such that for all continuous seminorms $p$ on $\mathcal{S}'_\beta$, the restrictions $Z^p_n$ Attouch-Wets converge to $Z^p$.
\item $Z_n$ converges to some function $Z : \mathcal{S}(\mathbb{R}^d) \to \mathbb{R}$ uniformly on bounded sets.
\end{enumerate}
Moreover, in this case $Z$ is the moment-generating function of $\omega$.
\end{corollary}
\begin{proof}
$(i) \implies (ii)$: By corollary \ref{cor:WeakConvergenceBanach}, for all continuous seminorms $p$ on $\mathcal{S}'$, $Z_n^p$ Mosco-converge to $Z^p$ for some convex and continuous function $Z : \mathcal{S} \to \mathbb{R}$.
In particular, $Z_n^p$ also converge to $Z^p$ in the Painlevé-Kuratowski sense which by \cite[Corollary 2.3]{src:EpigraphicalAndUniformConvergence} implies the uniform convergence on compact subsets of $(\mathcal{S}'_p)^*$.
Applying lemma \ref{lem:BanachCompactConvImpliesBoundedConv} we also obtain uniform convergence on bounded subsets of $(\mathcal{S}'_p)^*$.
Furthermore, every bounded subset of $\mathcal{S}$ is precompact, such that $Z^p$ is bounded on bounded subsets of $(\mathcal{S}'_p)^*$ which implies the Attouch-Wets convergence by \cite[Lemma 1.4]{src:EpigraphicalAndUniformConvergence}.

$(ii) \implies (iii)$: Every bounded subset $B \subset \mathcal{S}$ induces a continuous seminorm $p$ on $\mathcal{S}'_\beta$ such that $B$ is a bounded subset of $( \mathcal{S}'_p )^*$.
Furthermore, $Z$ is bounded on $B$ by the precompactness of $B$ in $\mathcal{S}$.
Hence, \cite[Corollary 2.2.]{src:EpigraphicalAndUniformConvergence} implies the uniform convergence on $B$.

$(iii) \implies (i)$: Let $p$ be a continuous seminorm on $\mathcal{S}'_\beta$.
By the pointwise convergence, we clearly have that $Z_n^p$ (M1)-converges to $Z^p$ while the continuity of $Z^p$ follows from lemma \ref{lem:ConvexLowerSemicontinuityConvergenceImpliesContinuous}.
For the (M2)-convergence, note that any sequence $(\phi_n)_{n \in \mathbb{N}}$ in $(\mathcal{S}'_p)^*$ converging weakly to some $\phi \in (\mathcal{S}'_p)^*$ is bounded and consequently also bounded in $\mathcal{S}$.
Also, by assumption, $Z_n^p$ converges to $Z^p$ uniformly on bounded sets.
Now, keeping in mind that $Z^p$ is also convex and continuous and thus weakly lower semicontinuous by \cite[Theorem 2.2.1]{src:Zalinescu:ConvexAnalysisInGeneralVectorSpaces}, we obtain
\begin{equation}
\liminf_{n \to \infty} Z_n^p \left( \phi_n \right)
\ge
\liminf_{n \to \infty} \left[ Z_n^p \left( \phi_n \right) - Z^p \left( \phi_n \right) \right]
+ \liminf_{n \to \infty} Z^p \left( \phi_n \right)
\ge Z^p \left( \phi \right) \, .
\end{equation}
Thus $Z_n^p$ Mosco-converges to $Z^p$ and corollary \ref{cor:WeakConvergenceBanach} applies.
\end{proof}
It would now be most tempting to conclude that there is some sort of Attouch-Wets topology on $\mathcal{S}$ with respect to which $Z_n$ converges to $Z$.
However, to the knowledge of the author, there have not been many studies of generalisations of Attouch-Wets convergence to non-normed spaces.
Consequently, no such result is available at the time of this writing.
Before coming to the final theorem stating the weak convergence of measures in terms of Attouch-Wets convergence of conjugate functions, we need another short lemma.
\begin{lemma}
\label{lem:ConjugateCharacterisationOfContinuity}
Let $f : \mathcal{S} \to \bar{\mathbb{R}}$ be proper convex and lower semicontinuous.
Then $f$ is a continuous function $\mathcal{S} \to \mathbb{R}$ if and only if its convex conjugate (Legendre-Fenchel transform) $f^c$ is supercoercive.
\end{lemma}
\begin{proof}
Let $B \subseteq \mathcal{S}(\mathbb{R}^d)$ be balanced, bounded and hence precompact.
Hence, by Fenchel-Moreau,
\begin{equation}
\label{eq:Supercoercivity}
\sup_{\phi \in B} f \left( \phi \right) = \sup_{\phi \in B, T \in \mathcal{S}'_\beta} \left[ T \left( \phi \right) - f^c \left( T \right) \right]
=
\sup_{T \in \mathcal{S}'_\beta} \left[ p_B \left( T \right) - f^c \left( T \right) \right] \, ,
\end{equation}
where $p_B$ is the continuous seminorm induced by $B$.

$\Rightarrow$:
Since all continuous seminorms $p$ are bounded by some $p_B$ the implication is clear.

$\Leftarrow$:
Equation \ref{eq:Supercoercivity} shows that $f$ is bounded on bounded subsets of $\mathcal{S}$.
Hence, it is finite everywhere and theorem \ref{thm:ConvexAndLowerSemicontinuousImpliesContinuous} applies.
\end{proof}
With this lemma at hand, we arrive at the final theorem of this section with immediate relevance for the Wetterich equation.
\begin{theorem}
\label{thm:ConvergenceTheorem}
Let $( \omega_n )_{n \in \mathbb{N}}$ be a sequence of Borel probability measures on $\mathcal{S}'_\beta$, $( Z_n )_{n \in \mathbb{N}}$ the corresponding moment-generating functions and set $W_n = \ln \circ Z_n$.
Then the following are equivalent:
\begin{enumerate}[label=(\roman*)]
\item $\limsup_{n \to \infty} W_n ( \phi ) < \infty$ for all $\phi \in \mathcal{S}$ and $\omega_n$ converges weakly to another Borel probability measure $\omega$.
\item There exists a proper convex, lower semicontinuous and supercoercive function $\Gamma : \mathcal{S}'_\beta \to \bar{\mathbb{R}}$ such that for all continuous seminorms $p$ on $\mathcal{S}'_\beta$ the lower semicontinuous envelopes $LSC(W_n^c,\iota_p)$ Attouch-Wets converge to $LSC(\Gamma,\iota_p)$, where $\iota_p : \mathcal{S}'_\beta \to \mathcal{S}'_p$ denotes the natural map.
\end{enumerate}
\end{theorem}
\begin{proof}
Let us first note that all $W_n$ are convex because all $Z_n$ are logarithmically convex.
Furthermore, all $\omega_n$ are probability measures such that $Z_n$ does not attain the value $0$ and $W_n$ does not attain the value $-\infty$.
Moreover, all $W_n$ are lower semicontinuous by the monotony of the logarithm and, by definition, $W_n ( 0 ) = 0$.
Thus, letting $p$ denote some continuous seminorm on $\mathcal{S}'_\beta$, the restrictions $W_n^p$ are also proper convex and lower semicontinuous functions.
Recalling that the Legendre-Fenchel transform is a bijection between proper convex and lower semicontinuous functions on $\mathcal{S}'_p$ and $(\mathcal{S}'_p)^*$, there exist proper convex and lower semicontinuous functions $\Gamma_{n,p} : \mathcal{S}'_p \to \mathbb{R}$ such that
$\Gamma_{n,p}^c = W_n^p$ for all $n \in \mathbb{N}$.
By Fenchel-Moreau $\Gamma_{n,p}$ is equal to the restriction of $(W_n^p)^c$ to $\mathcal{S}'_p$.
Furthermore, $\limsup_{n \to \infty} W_n (\phi) < \infty$ implies $\limsup_{n \to \infty} Z_n (\phi) < \infty$.

$\Rightarrow$: By corollary \ref{cor:WeakConvergenceUniform}, $Z_n$ converges uniformly on bounded sets to the moment-generating function $Z$ of $\omega$ which is also continuous and never zero because $\omega$ is a probability measure.
On bounded sets $Z$ is always bounded away from zero by the precompactness of such sets.
Consequently, $W$ is continuous, $W_n$ also converges uniformly to $W$ on bounded sets and the same is true for the restrictions, i.e. $W_n^p$ converges uniformly to $W^p$ on bounded sets.
Applying \cite[Lemma 1.4]{src:EpigraphicalAndUniformConvergence}, we have that $W_n^p$ Attouch-Wets converges to $W^p$.

Recalling that the Legendre-Fenchel transform is a homeomorphism with respect to Attouch-Wets convergence \cite{src:Back:ContinuityOfConjugation}, it is clear that $\Gamma_{n,p}$ Attouch-Wets converges to the restriction of $W_p^c$ to $\mathcal{S}'_p$.
Now, lemma \ref{lem:LSCConjugation} shows that $W_p^c$ and $LSC(W^c,\iota_p)$ agree on $\mathcal{S}'_p$.
That $W^c$ is supercoercive is clear from lemma \ref{lem:ConjugateCharacterisationOfContinuity}.

$\Leftarrow$:
By the homeomorphism property of the Legendre-Fenchel transform and lemma \ref{lem:LSCConjugation}, $W_n^p$ Attouch-Wets converges to $(\Gamma^c)^p$.
Furthermore, $\Gamma^c$ is continuous by lemma \ref{lem:ConjugateCharacterisationOfContinuity} such that $(\Gamma^c)^p$ is bounded on bounded sets by the precompactness of bounded sets in $\mathcal{S}$.
Hence, $W_n^p$ actually converges to $\Gamma^c$ uniformly on bounded sets \cite[Corollary 2.2]{src:EpigraphicalAndUniformConvergence}.
Because every bounded set $B$ in $\mathcal{S}$ induces a continuous seminorm $q$ on $\mathcal{S}'_\beta$ such that $B$ is also bounded in $(\mathcal{S}'_q)^*$, it follows that $W_n$ converges to $\Gamma^c$ uniformly on bounded sets.
Finally, since bounded sets in $\mathcal{S}(\mathbb{R}^d)$ are precompact, $\Gamma^c$ is bounded on bounded sets such that $Z_n$ converges to $\exp [ \Gamma^c ]$ uniformly on bounded sets.
Hence, corollary \ref{cor:WeakConvergenceUniform} applies.
\end{proof}
\begin{remark}
While the convergence criterion for the dual functions $W_n^c$ is somewhat complicated, the concrete form of $LSC(\Gamma,\iota_p)$ is rather simple and may be extracted from lemma \ref{lem:LSCEnvelopeSupercoercive} due to the supercoercivity of $\Gamma$.
In the case where $p$ is a continuous norm on $\mathcal{S}'_\beta$, $LSC(\Gamma,\iota_p)$ is in fact simply given by $\Gamma$ on $\mathcal{S}'_\beta$ considered as a subspace of $\mathcal{S}'_p$ and $\infty$ everywhere else.
Likewise, if the regularisation happens to be such that all $W_n$ are continuous, the resulting supercoercivity also gives simple expressions for $W_n^c$.
As seen in section \ref{sec:WetterichDerivation}, this is indeed the case for the presented regularisation scheme.
Even in these cases, it appears to the author that there is no straightforward simplification of the dual convergence criterion.
The reason is that Attouch-Wets convergence in the spaces $\mathcal{S}'_p$ does not provide enough uniformity for a sequential statement such as (M1)-convergence in $\mathcal{S}'_\beta$ because the latter is not a Fréchet-Urysohn space.
\end{remark}
The objects $W_n^c$ in the above theorem correspond to a set of quantum effective actions of regularised theories and are precisely those objects which may be computed from the Wetterich equation.
If they are ordered such that any regularisation vanishes in the $n \to \infty$ limit, theorem \ref{thm:ConvergenceTheorem} states necessary and sufficient conditions for their convergence to a full theory in a physically meaningful manner according to the discussion in the beginning of this section.
\section{A Functional Regularisation Scheme}
\label{sec:FunctionalRegularisationScheme}
The goal of Quantum Field Theory à la Osterwalder-Schrader can be phrased in making sense of Euclideanised Feynman integrals.
A typical example is the moment-generating function
\begin{equation}
Z \left( J \right) = \int \exp \left[ -S \left( \psi \right) + \int J \psi \right] \mathrm{d} \psi \, ,
\end{equation}
where
\begin{itemize}
\item the integral is taken over some path space,
\item the measure is taken to be translation invariant,
\item $J$ lies in a dual space of the configuration space and
\item $S$ is the Euclideanised classical action typically containing terms like $\int \psi^4$.
\end{itemize}
In view of the Osterwalder-Schrader theorem one would like the integral to run over $\mathcal{S}'_\beta$, which
\begin{itemize}
\item is infinite-dimensional,
\item admits no non-trivial translation invariant measure and
\item makes expressions like $\psi^4$ generally ill-defined.
\end{itemize}
Hence, one typically regularises the path space in some fashion e.g by restricting to a finite volume and enforcing a momentum cutoff, rendering the regularised path space finite-dimensional \cite[Chapter 8]{src:ReuterSaueressig}.
As outlined in the last section the resulting generating functions $Z$ respectively the corresponding measures are then interpreted as to reflect real physics to a degree where the regularisation is assumed to introduce only negligible effects, e.g
\begin{itemize}
\item physics in a volume much smaller than the finite volume introduced by the regularisation,
\item physical phenomena with characteristic momentum scales much smaller than the enforced cutoff.\footnote{These conditions are debatable in view of the Euclideanised theory not living on physical Minkowski space.}
\end{itemize}
These particular regularisations define a directed subset of $\mathbb{R} \times \mathbb{R}$ where
\begin{equation}
\left( V, p_{\mathrm{max}} \right) \le \left( V', p'_{\mathrm{max}} \right)
\qquad
\Leftrightarrow
\qquad
V \le V' \quad \text{and} \quad p_{\mathrm{max}} \le p'_{\mathrm{max}} \, .
\end{equation}
Hence, all correspondingly regularised measures can be collected into a net $\omega_{V, p_{\mathrm{max}}}$ and we may interpret the requirement that they should describe \emph{the same physics} within their respective window of applicability as a kind of gluing instruction.
But this implies that for ever larger volumes and momentum cutoffs one expects to approach a limit that encompasses \emph{all} physical phenomena.
If this limit is indeed a Quantum Field Theory, we may phrase this requirement as the existence of a measure $\omega$ satisfying the Osterwalder-Schrader axioms and
\begin{equation}
\label{eq:VolumeMomentumLimit}
\lim_{V, p_{\mathrm{max}} \to \infty} \omega_{V, p_{\mathrm{max}}} = \omega
\end{equation}
in some appropriate sense.
As is well-known, the existence of such a limit generally requires \emph{renormalising} $S$, e.g in the case of having a term $\lambda \int \psi^4$, promoting $\lambda \in \mathbb{R}$ to a function depending on $V$ and $p_{\mathrm{max}}$.

It is quite clear that the choice of regularisation scheme immensely affects concrete calculations.
For instance, the finite volume setting is typically implemented as a compactification of $\mathbb{R}^d$ to a torus leading to discrete eigenvalues of the Laplacian, i.e discrete admissible momenta.
The personal belief of the author, however, is that a regularisation scheme maintaining as much smoothness as possible is desirable for computational and analytic methods alike.
One such scheme is the following, where the limit $n \to \infty$ corresponds to the removal of the regularistion:
\begin{enumerate}
\item For each $n \in \mathbb{N}$ split $S_n$ into a free part $S_n^{\mathrm{free}}$ and an interacting part $S_n^{\mathrm{int}}$ such that
$S_n^{\mathrm{free}} : \mathcal{S} \to \mathbb{R}$,
\begin{equation}
\phi \mapsto \left( \phi, B_n \phi \right)
\end{equation}
for some continuous, bijective, linear operator $B_n$ on $\mathcal{S}$ which is bounded from below in the sense that there is an $\eta_n > 0$ such that for all $\phi \in \mathcal{S}$
\begin{equation}
\label{eq:FreePartIsL2BoundedFromBelow}
\left( \phi, B_n \phi \right) \ge \eta_n \left( \phi, \phi \right) \, .
\end{equation}
The prototypical example is of course $B_n = m^2 - \Delta$ for a free, scalar, massive field theory on $\mathbb{R}^d$.
Define the corresponding centred Radon Gaußian probability measure $\mu_n$ on $\mathcal{S}'_\beta$ given by its characteristic function
\begin{equation}
\hat{\mu}_n \left( \phi \right) = \exp \left[ - \frac{1}{2} \left( \phi, B_n^{-1} \phi \right) \right] \, .
\end{equation}
for all $\phi \in \mathcal{S}$.
Then, $\mu_n$ encodes the free theory determined by $B_n$.
\item Fix a sequence $(\mathcal{R}_n)_{n \in \mathbb{N}}$ of linear and continuous operators $\mathcal{S}'_\beta \to \mathcal{S}$ with
\begin{equation}
\lim_{n \to \infty} \left( \iota \mathcal{R}_n \iota \right) \left( \phi \right) \left( \psi \right)
=
\lim_{n \to \infty} \left( \mathcal{R}_n \iota \phi, \psi \right)
=
\left( \phi, \psi \right) = \iota \left( \phi \right) \left( \psi \right)
\end{equation}
for all $\phi, \psi \in \mathcal{S}$.
A simple example is given by choosing
\begin{equation}
\begin{aligned}
\chi_n \left( x \right) &= \exp \left[ - \frac{1}{2} \frac{\left \Vert x \right \Vert^2}{n^2 K^2} \right] \\
\xi_n \left( x \right) &= \left( \frac{n^2 \Lambda^2}{2 \pi} \right)^{d/2} \exp \left[ - \frac{n^2 \Lambda^2}{2} \left \Vert x \right \Vert^2 \right]
\end{aligned}
\end{equation}
for some $K, \Lambda > 0$, all $x \in \mathbb{R}^d$ and setting $\mathcal{R}_n T = \chi_n \cdot ( \xi_n \ast T )$ where $\ast$ denotes convolution of functions.
Then, $n \Lambda$ can be viewed as a momentum cutoff and $n K$ as the radius of the ball to which we restrict the theory.
The contributions from outside these physical windows are heavily suppressed by the exponential decays of $\chi_n$ and $\xi_n$ respectively.

Furthermore, these particular $\mathcal{R}_n$ commute with $\mathcal{O}(\mathbb{R}^d)$ such that the rotational invariance is kept intact.
They also satisfy the stronger condition that $\lim_{n \to \infty} \iota \circ \mathcal{R}_n = \mathrm{id}$ uniformly on bounded sets which may be beneficial in specific applications.
However, it can be shown that they break reflection positivity as is explained in section \ref{sec:ReflectionPositivity}.

\item For all $n \in \mathbb{N}$ define the Radon probability measures $\nu_n = (\mathcal{R}_n)_* \mu_n$ and
\begin{equation}
\label{eq:omegaNDefinition}
\omega_n = \iota_\ast \left( \frac{\exp \left[ -S_n^{\mathrm{int}} \right]}{\int \exp \left[ -S_n^{\mathrm{int}} \right] \mathrm{d} \nu_n} \cdot \nu_n \right)
\end{equation}
on $\mathcal{S}$ and $\mathcal{S}'_\beta$ respectively.
\end{enumerate}
In analogy to equation \ref{eq:VolumeMomentumLimit}, we now have a sequence $(\omega_n)_{n \in \mathbb{N}}$ of measures corresponding to regularised theories and the ultimate goal is to find a limit $\omega$  satisfying the Osterwalder-Schrader axioms.
\begin{remark}
The use of the $L^2(\mathbb{R}^d)$ inner product in equations \ref{eq:FreePartIsL2BoundedFromBelow} and \ref{eq:omegaNDefinition} could be generalised to any other continuous inner product on $\mathcal{S}$.
However, in most applications the $L^2(\mathbb{R}^d)$ one suffices and is the simplest to work with such that there is hardly any practical benefit in admitting more general structures.
\end{remark}
\begin{remark}
To demand that $\exp [ - S_n^{\mathrm{int}} ] \in L^1 ( \nu_n )$ is strictly necessary for this regularisation scheme to work.
As a matter of fact, this requirement is not trivial e.g. with regard to mass counterterms in the $\phi^4_4$ modes.
It has, however, recently been proven in \cite[Theorem 1.1]{src:HJZ} that it is indeed satisfied.
\end{remark}

\section{A Note on Reflection Positivity}
\label{sec:ReflectionPositivity}
\begin{definition}
Let $\mathcal{S}_+ \subseteq \mathcal{S}$ denote the set of Schwartz functions with support in $\mathbb{R}_{\ge 0} \times \mathbb{R}^{d-1}$ equipped with its subspace topology turning it into a reflexive nuclear space.
For any $\phi \in \mathcal{S}$, let $\theta \phi \in \mathcal{S}$ with
\begin{equation}
\left( \theta \phi \right) \left( x_1, x_2, \dots, x_d \right) = \phi \left( -x_1, x_2, \dots, x_d \right)
\end{equation}
for all $x \in \mathbb{R}^d$.
Following \cite{src:GlimmJaffe}, we call a finite Borel measure $\mu$ on $\mathcal{S}'_\beta$ \textbf{reflection positive} if
\begin{equation}
\sum_{m, n = 1}^N c_m^* c_n \hat{\mu} \left( \phi_n - \theta \phi_m \right) \ge 0
\end{equation}
for all $N \in \mathbb{N}$ as well as all sequences $(c_n)_{n \in \mathbb{N}}$ in $\mathbb{C}$ and $( \phi_n )_{n \in \mathbb{N}}$ in $\mathcal{S}_+$.
\end{definition}
One can show that the simple example of $\mathcal{R}_n$ given in the last section does not preserve reflection positivity in the sense that the resulting non-interacting measures $\iota_* \nu_n$ are not reflection positive in general.
It is, however, possible to retain reflection possitivity if one is willing to break $\mathcal{O}(d)$ covariance:
\begin{example}
First, define
\begin{equation}
\Phi \left( x \right)
=
\exp \left[ -\frac{1}{2} \left \Vert x \right \Vert^2 \right] 
\qquad\text{and}\qquad
\Psi \left( x \right) = \begin{cases}
\exp \left[ - \frac{1}{x_1} \right] & x_1 > 0 \\
0 & \mathrm{else}
\end{cases} \, .
\end{equation}
Then, fix some $K, \Lambda, M > 0$ and set
\begin{equation}
\kappa_n \left( x \right) = \Psi \left( n M x \right)
\, \text{,} \quad
\xi_n(x) = \left( n \Lambda M \right)^d \frac{\Phi \left( n \Lambda x \right) \Psi(n M x)}{\left \Vert \Phi \Psi \right \Vert_{L^1(\mathbb{R}^d)}}
\, \text{,} \quad
\chi_n \left( x \right) = \Phi \left( \frac{x}{n K} \right)
\end{equation}
for all $x \in \mathbb{R}^d$ and all $n \in \mathbb{N}$.
Now, let
\begin{equation}
\mathcal{R}_n T = \kappa_n \cdot \chi_n \cdot \left( \theta \xi_n \ast T \right) + \theta \kappa_n \cdot \chi_n \cdot \left( \xi_n \ast T \right)
\end{equation}
for all $T \in \mathcal{S}'_\beta$.
Then, for any sequence $(\mu_n)_{n \in \mathbb{N}}$ of reflection positive finite Radon measures on $\mathcal{S}'_\beta$, the measures $\omega_n = \iota_* (\mathcal{R}_n)_* \mu_n$ are reflection positive,
\begin{equation}
\lim_{n \to \infty} \left( \iota \circ \mathcal{R}_n \circ \iota \right) \left( \phi \right)  = \iota(\phi)
\end{equation}
for all $\phi \in \mathcal{S}$ and $\mathcal{R}_n$ commutes with spatial $\mathcal{O}(d-1)$ rotations.
\end{example}
\begin{proof}
That $\mathcal{R}_n$ commutes with spatial $\mathcal{O}(d-1)$ rotations follows directly from the corresponding invariance of $\Phi$ and $\Psi$.
Furthermore, it is straightforward to verify that
\begin{equation}
\ltrans{\mathcal{R}}_n \iota \phi = \xi_n \ast \left( \chi_n \cdot \kappa_n \cdot \phi \right) + \theta \xi_n \ast \left( \chi_n \cdot \theta \kappa_n \cdot \phi \right)
\end{equation}
for all $\phi \in \mathcal{S}$ from which it directly follows that $\ltrans{\mathcal{R}}_n \circ \iota$ commutes with $\theta$ by the $\mathcal{O}(d)$ invariance of $\chi_n$.
Consequently, for any $\phi, \psi \in \mathcal{S}$,
\begin{equation}
\hat{\omega}_n \left( \phi - \theta \psi \right)
=
\hat{\mu}_n \left( \ltrans{\mathcal{R}}_n \iota \phi - \theta \ltrans{\mathcal{R}}_n \iota \psi \right) \, .
\end{equation}
Now, reflection positivity follows if $\ltrans{\mathcal{R}}_n \circ \iota$ preserves the property of being supported in $\mathbb{R}_{\ge 0} \times \mathbb{R}^{d-1}$.
This is easily seen since
\begin{equation}
\begin{aligned}
\mathrm{supp} \, \ltrans{\mathcal{R}}_n \iota \phi
&=
\mathrm{supp} \, \xi_n \ast \left( \chi_n \cdot \kappa_n \cdot \phi \right) \\
&\subseteq
\mathbb{R}_{\ge 0} \times \mathbb{R}^{d-1}
\, + \,
\mathbb{R}_{\ge 0} \times \mathbb{R}^{d-1}
=
\mathbb{R}_{\ge 0} \times \mathbb{R}^{d-1}
\end{aligned}
\end{equation}
for any $\phi \in \mathcal{S}$ supported in $\mathbb{R}_{\ge 0} \times \mathbb{R}^{d-1}$.

To see the convergence, note that for any $\phi, \psi \in \mathcal{S}$
\begin{equation}
\left( \iota \circ \mathcal{R}_n \circ \iota \right) \left( \phi \right) \left( \psi \right)
=
\int_{\mathbb{R}^d}
\left[
\kappa_n \cdot \chi_n \cdot \left( \theta \xi_n * \phi \right) \cdot \psi
+
\theta \kappa_n \cdot \chi_n \cdot \left( \xi_n * \phi \right) \cdot \psi
\right] \, .
\end{equation}
Finally, it is clear that
\begin{itemize}
\item $\xi_n * \phi \to \phi$ and $\theta \xi_n * \phi \to \phi$ in $L^2(\mathbb{R}^d)$,
\item $\kappa_n \cdot \chi_n \cdot \psi \to \psi \cdot I ( \mathbb{R}_{> 0} \times \mathbb{R}^{d-1} )$ in $L^2(\mathbb{R}^d)$,
\item $\theta \kappa_n \cdot \chi_n \cdot \psi \to \psi \cdot I ( \mathbb{R}_{< 0} \times \mathbb{R}^{d-1} )$ in $L^2(\mathbb{R}^d)$,
\end{itemize}
where $I(A)$ denotes the indicator function of a set $A \subseteq \mathbb{R}^d$.
Consequently,
\begin{equation}
\lim_{n \to \infty}
\left( \iota \circ \mathcal{R}_n \circ \iota \right) \left( \phi \right) \left( \psi \right)
=
\int_{\mathbb{R}^d}
\phi \psi
= \iota \left( \phi \right) \left( \psi \right) \, .
\end{equation}
\end{proof}
\begin{remark}
Note that we do not have that $\iota \circ \mathcal{R}_n \to \mathrm{id}$.
Indeed, letting $\delta_0$ be the Dirac distribution at zero, we have $\mathcal{R}_n \delta_0 = 0$ for all $n \in \mathbb{N}$.
\end{remark}
It is thus possible to define the operators $\mathcal{R}_n$ in a way that the resulting Gaussian measures are reflection positive.
When combining a reflection positive Gaussian measure with a density, some care must be taken to ensure that reflection positivity is maintained.
For some models on a lattice this is e.g. discussed in \cite{src:GlimmJaffe}.
\section{The Functional Renormalisation Group Equation}
\label{sec:WetterichDerivation}
In this section we are going to use the introduced regularisation scheme to define a regularised Quantum Field Theory with a slightly modified classical action given by adding a bilinear operator $F^n_k$ to $S_n^{\mathrm{int}}$.

In particular, every object in this section is taken to carry a regularisation index $n \in \mathbb{N}$ but for legibility it will not be written explicitly.

Letting $\nu$ and $S^{\mathrm{int}}$ be as in section \ref{sec:FunctionalRegularisationScheme}, we first demand the following regularity properties:
\begin{itemize}
\item $S^{\mathrm{int}} : \mathcal{S} \to \mathbb{R}$ is continuous,
\item there is a $q > 1$ such that $\exp [-S^{\mathrm{int}}] \in L^q(\nu)$,
\item there is a continuous seminorm $p$ on $\mathcal{S}$ and $C > 0$ such that $\exp [-S^{\mathrm{int}}] \le C \exp[ p^2 ]$.
\end{itemize}
\begin{remark}
The continuity is important for the boundary condition at $k \to \infty$ and it also ensures the strict positivity of $\exp [-S^{\mathrm{int}}]$ as well as its boundedness on compacta.
The second condition is only slightly stronger than $\exp [-S^{\mathrm{int}}] \in L^1(\nu)$, which is necessary for the regularisation scheme to work and enables the use of Hölder's inequality.
The third condition is used in theorem \ref{thm:DiracMeasureApproximation} enabling a weak continuity of a suitable approximation of a Dirac measure.
It is possible that it can be relaxed significantly.
\end{remark}
Let us denote the Cameron-Martin space of $\nu$ by $H ( \nu ) \subset \mathcal{S}$ and the reproducing kernel Hilbert space of $\nu$ by $\mathcal{S}'_\nu$ with its inner product $\left< \cdot, \cdot \right>$.
The spaces $H ( \nu )$ and $\mathcal{S}'_\nu$ will be of central importance in the derivation of the FRGE.
In fact, much of the proof could be given in a more abstract setting and essentially follows from properties of Gaußian measures on locally convex spaces.

Now, let $( F_k )_{k \in \mathbb{R}}$ be a family of symmetric bilinear operators on $\mathcal{S}$ with the following properties:
\begin{itemize}
\setlength\itemsep{0.5em}
\label{list:FProperties}
\item \emph{For easy differentiation, allow negative values of $k$}:
\begin{itemize}[label=$\bullet$]
\item $\forall \, \phi \in \mathcal{S}, k < 0 : \quad F_k \left( \phi, \phi \right) = 0$,
\end{itemize}
\item \emph{For Dirac delta measure approximation (see theorem \ref{thm:DiracMeasureApproximation})}:
\begin{itemize}[label=$\bullet$]
\setlength\itemsep{0.5em}
\item $\forall \, \phi \in \mathcal{S}, k \ge 0 : \quad 0 \le F_k \left( \phi, \phi \right) \le  k^2\left( \phi, \phi \right)$, \hspace*{\fill}\tagx[itm:FkBound]
\item $\forall \, \phi \in \mathcal{S} \, \exists \, C, K > 0 \, \forall \, k \ge K: \quad F_k \left( \phi, \phi \right) \ge C k^2 \left( \phi, \phi \right)$, \hspace*{\fill}\tagx[itm:FkDivergence]
\end{itemize}
\item \emph{For differentiability in lemmas \ref{lem:ZDifferentiability} and \ref{lem:YDifferentiability}}:
\begin{itemize}[label=$\bullet$]
	\setlength\itemsep{0.5em}
	\item $F$ is pointwise continuously $k$-differentiable, i.e for all $k \in \mathbb{R}$ and $\phi \in \mathcal{S}$
	\begin{equation}
	F'_k \left( \phi, \phi \right) := \lim_{t \to 0} \frac{F_{k+t} \left( \phi, \phi \right) - F_k \left( \phi, \phi \right)}{t}
	\end{equation}
	exists and is jointly continuous in $k$ and $\phi$,
	\item The above convergence is uniform in $\phi$ in the sense that there is a continuous seminorm $p$ on $\mathcal{S}$ such that for all $k \in \mathbb{R}$ and all $\epsilon > 0$ there exists some $\delta > 0$ as well as a function $o : ( - \delta, \delta ) \to \mathbb{R}$ such that $\lim_{t \to 0} o ( t )/t = 0$ and
	\begin{equation}
	\left \vert F_{k+t} \left( \phi, \phi \right) - F_k \left( \phi, \phi \right) - t F'_k \left( \phi, \phi \right) \right \vert < \epsilon o \left( t \right) p \left( \phi \right)^2
	\end{equation}
	for all $t \in ( - \delta, \delta )$ and $\phi \in \mathcal{S}$,
\end{itemize}
\item \emph{For positivity and interchange of integrations in theorem \ref{thm:FRGE}}:
\begin{itemize}[label=$\bullet$]
	\item For all $k \in \mathbb{R}$ there is a $\sigma$-finite measure space $( X^k, \mathcal{A}^k, m_k )$ and a mapping $U^k : X^k \to ( \mathcal{S}'_\nu)_{\mathbb{C}}$ such that $X^k \to \mathbb{R}, x \mapsto U^k_x ( \phi )$ is $m_k$-measurable for all $\phi \in \mathcal{S}$ and
	\begin{equation}
	\label{eq:FKDerivative}
	F'_k \left( \phi, \phi \right) = \int_{X^k} \left \vert U^k_x \left( \phi \right) \right \vert^2 \mathrm{d} m_k \left( x \right) \, .
	\end{equation}
\end{itemize}
\end{itemize}
While especially the last condition looks very technical, common choices (note the absence of prefactors for separate wave-function renormalisation) like
\begin{itemize}
\setlength\itemsep{0.5em}
\item $F_k \left( \phi, \phi \right)= \int_{\mathbb{R}^d} \left \vert \tilde{\phi} \left( p \right) \right \vert^2 \left( k^2 - \left \Vert p \right \Vert^2 \right) \theta \left( k^2 - \left \Vert p \right \Vert^2 \right) \mathrm{d} p$
\hfill a.k.a the Litim regulator,
\item $F_k \left( \phi, \phi \right)= \int_{\mathbb{R}^d} \left \vert \tilde{\phi} \left( p \right) \right \vert^2 \frac{\left \Vert p \right \Vert^2}{\exp \left[ \frac{ \left \Vert p \right \Vert^2}{k^2} \right] - 1} \mathrm{d} p$ \hfill a.k.a the exponential regulator,
\end{itemize}
for $k > 0$ are included which in particular do not carry any extra dependence on $n \in \mathbb{N}$.
\begin{remark}
Just as one could generalise to other continuous inner products in equation \ref{eq:FreePartIsL2BoundedFromBelow}, one can generalise the upper and lower $F_k$ bounds in equations \ref{itm:FkBound} and \ref{itm:FkDivergence} to other continuous seminorms.
But since the established choices of $F_k$ all work with the $L^2(\mathbb{R}^d)$ inner product there is little practical reason to do so.
\end{remark}
For brevity, define
\begin{align}
N_k &= \int_S \exp \left[ -S^{\mathrm{int}} \left( \psi \right) - \frac{1}{2} F_k \left( \psi, \psi \right) \right] \mathrm{d} \nu \left( \psi \right) \, , \\
f_k \left( \phi \right) &=
\exp \left[ -S^{\mathrm{int}} \left( \phi \right) - \frac{1}{2} F_k \left( \phi, \phi \right) \right]
\end{align}
for all $k \in \mathbb{R}$ and $\phi \in \mathcal{S}$.
Then, we may consider the family $\{ f_k / N_k \cdot \nu : k \in \mathbb{R}  \}$ of probability measures on $\mathcal{S}$ and in view of the last section the object of interest is $f_0 / N_0 \cdot \nu$.
By the properties of $S^{\mathrm{int}}$ we clearly have that $f_k : \mathcal{S} \to \mathbb{R}$ is strictly positive and there exists a $q \in ( 1, \infty ]$ such that $f_k \in L^q ( \nu )$ for all $k \in \mathbb{R}$.
Let us now define a family $Z : \mathbb{R} \times \mathcal{S}'_\nu \to \mathbb{R}$ of moment-generating functions as
\begin{equation}
Z_k \left( T \right)
=
\frac{1}{N_k} \int_{\mathcal{S}} \exp \left[ T \left( \psi \right) \right] f_k \left( \psi \right) \mathrm{d} \nu \left( \psi \right) \, .
\end{equation}
By employing the Cameron-Martin theorem, we have
\begin{equation}
\label{eq:ZRepresentation2}
Z_k \left( T \right) = \frac{1}{N_k} \exp \left[ \frac{1}{2} \left \langle T, T \right \rangle \right] \int_{\mathcal{S}} f_k \left( \psi + R_\nu T \right) \mathrm{d} \nu \left( \psi \right)
\end{equation}
for all $k \in \mathbb{R}$ and $T \in \mathcal{S}'_\nu$ such that $Z$ is indeed well-defined (everywhere finite).
Also, by virtue of \cite[Theorem 2.4.8]{src:Bogachev:GaußianMeasures} each $Z_k : \mathcal{S}'_\nu \to \mathbb{R}$ is continuous.
A straightforward calculation - that we shall omit here - shows that one may differentiate under the integral sign:
\begin{lemma}
\label{lem:ZDifferentiability}
$Z$ is continuously Fréchet differentiable and its derivative at $( k, J )$ is given by
\begin{equation}
D_{k, T} Z = \begin{pmatrix}
- \frac{1}{2 N_k} \int_{\mathcal{S}} F'_k \left( \psi, \psi \right) \exp \left[ T \left( \psi \right) \right] f_k \left( \psi \right) \mathrm{d} \nu \left( \psi \right)
- Z_k \left( T \right) \partial_k \ln N_k
\\
\frac{1}{N_k} \int_{\mathcal{S}} \psi \exp \left[ T \left( \psi \right) \right] f_k \left( \psi \right) \mathrm{d} \nu \left( \psi \right)
\end{pmatrix}
\end{equation}
where the term in the second row ($D_T Z_k$) may be understood as a (generalised) Bochner integral in $\mathcal{S}$ \cite[theorem 3]{src:ErikThomas:IntegrationInS} and in fact $D_T Z_k \in H ( \nu )$.
\end{lemma}
Let us also define $Y : \mathbb{R} \times \mathcal{S}'_\nu \to H(\nu)$ with $( k, T ) \mapsto D_T Z_k$.
\begin{lemma}
\label{lem:YDifferentiability}
$Y$ is continuously Fréchet differentiable and its derivative at $( k, J )$ is given by
\begin{equation}
\left( D_{k, J}  Y \right) \left( l, T \right) = \begin{pmatrix}
- \frac{l}{2 N_k} \int_{\mathcal{S}} \psi F'_k \left( \psi, \psi \right) \exp \left[ J \left( \psi \right) \right] f_k \left( \psi \right) \mathrm{d} \nu \left( \psi \right)
- l D_J Z_k \cdot \partial_k \ln N_k \\
\frac{1}{N_k} \int_{\mathcal{S}} \psi T \left( \psi \right) \exp \left[ J \left( \psi \right) \right] f_k \left( \psi \right) \mathrm{d} \nu \left( \psi \right)
\end{pmatrix}
\end{equation}
for all $l \in \mathbb{R}, T \in \mathcal{S}'_\nu$.
Both integrals may again understood to be generalised Bochner integrals in $\mathcal{S}$ with values in $H ( \nu )$.
\end{lemma}
These properties are inherited by $W : \mathbb{R} \times \mathcal{S}'_\nu \to \mathbb{R}, ( k, T ) \mapsto \ln Z_k ( T )$, i.e $W$ is continuously differentiable, $W_k$ is twice continuously differentiable and $D W_k : \mathcal{S}'_\nu \to H ( \nu )$.
As a matter of fact, $D W_k$ even turns out to be a bijection between $\mathcal{S}'_\nu$ and $H ( \nu )$.
The injectivity follows directly from the following positivity property of $D^2 W_k$.
\begin{theorem}
\label{thm:D2WkPositivity}
For all $k \in \mathbb{R}$ and $J \in \mathcal{S}'_\nu$ there exists a $C > 0$ such that for all $K \in \mathcal{S}'_\nu$
\begin{equation}
K \left[ \left( D^2_J W_k \right) \left( K \right) \right] \ge C \left \langle K, K \right \rangle \, .
\end{equation}
\end{theorem}
\begin{proof}
By Hölder's inequality,
\begin{equation}
\label{eq:D2WkEstimateFactorisation}
\begin{aligned}
K \left[ \left( D^2_J W_k \right) \left( K \right) \right] 
&=
\frac{1}{N_k^2 Z_k \left( J \right)^2}
\int_{\mathcal{S} \times \mathcal{S}}
\left[ K \left( \psi \right)^2 - K \left( \psi \right) K \left( \phi \right) \right] \\
&\phantom{=} \times
\exp \left[ J \left( \psi \right) + J \left( \phi \right) \right] f_k \left( \psi \right) f_k \left( \phi \right) \mathrm{d} \left( \nu \times \nu \right) \left( \psi, \phi \right) \\
&\ge
\frac{1}{N_k^2 Z_k \left( J \right)^2}
\int_{\mathcal{S} \times \mathcal{S}}
\left[ \left \vert K \left( \psi \right) K \left( \phi \right) \right \vert - K \left( \psi \right) K \left( \phi \right) \right] \\
&\phantom{\ge} \times
\exp \left[ J \left( \psi \right) + J \left( \phi \right) \right] f_k \left( \psi \right) f_k \left( \phi \right) \mathrm{d} \left( \nu \times \nu \right) \left( \psi, \phi \right) \\
&=
\frac{1}{N_k^2 Z_k \left( J \right)^2}
\int_{\mathcal{S}}
\left[ \left \vert K \left( \psi \right) \right \vert - K \left( \psi \right) \right]
\exp \left[ J \left( \psi \right) \right] f_k \left( \psi \right) \mathrm{d} \nu \left( \psi \right) \\
&\phantom{=} \times
\int_{\mathcal{S}}
\left[ \left \vert K \left( \psi \right) \right \vert + K \left( \psi \right) \right]
\exp \left[ J \left( \psi \right) \right] f_k \left( \psi \right) \mathrm{d} \nu \left( \psi \right)
\end{aligned}
\end{equation}
which clearly is nonnegative.
Let us now suppose that there is a sequence $( K_n )_{n \in \mathbb{N}}$ in $\mathcal{S}'_\nu$ with $\left< K_n, K_n \right> = 1$ such that the first integral tends to zero, i.e
\begin{equation}
\begin{aligned}
&\lim_{n \to \infty} \int_{\mathcal{S}}
\left[ \left \vert K_n \left( \psi \right) \right \vert - K_n \left( \psi \right) \right]
\exp \left[ J \left( \psi \right) \right] f_k \left( \psi \right) \mathrm{d} \nu \left( \psi \right) = 0 \, .
\end{aligned}
\end{equation}
Then, there exists a subsequence $( L_n )_{n \in \mathbb{N}}$ such that $\lim_{n \to \infty} \vert L_n ( \psi ) \vert - L_n ( \psi ) = 0$ for $\nu$-almost every $\psi$.
But since each $L_n$ can in turn be written as the $\nu$-almost everywhere pointwise limit of linear functions, the above implies $\lim_{n \to \infty} L_n ( \psi ) = 0$ for $\nu$-almost every $\psi$.
One arrives at the same conclusion if one takes the second integral to go to zero instead.
By the finiteness of $\nu$ we thus have that $L_n \to 0$ in $\nu$-measure as $n \to \infty$.
Now, let $\epsilon > 0$ and pick any $\nu$-measurable $A \subset \mathcal{S}$ with $\nu ( A ) < \epsilon^2 / 3$.
Then, for all $n \in \mathbb{N}$
\begin{equation}
\int_A L_n \left( \psi \right)^2 \mathrm{d} \nu \left( \psi \right)
\le
\sqrt{ \int_S L_n \left( \psi \right)^4 \mathrm{d} \nu \left( \psi \right) }
\sqrt{ \nu \left( A \right) } \\
<
\left \langle L_n, L_n \right \rangle \epsilon = \epsilon \, .
\end{equation}
Thus, Vitali's convergence theorem tells us that $L_n$ goes to zero in $L^2 ( \nu )$ i.e in $\mathcal{S}'_\nu$ for $n \to \infty$ which is a contradiction.
\end{proof}
\begin{corollary}
$D W_k : \mathcal{S}'_\nu \to H ( \nu )$ is injective for all $k \in \mathbb{R}$.
\end{corollary}
\begin{proof}
Suppose that $D_J W_k = D_K W_k$ for some $J, K \in \mathcal{S}'_\nu$.
Then, by Rolle's theorem, there is a $t \in [ 0, 1 ]$ such that
\begin{equation}
\left( J - K \right) \left[ \left( D^2_{t J + \left( 1 - t \right) K} W_k \right) \left( J - K \right) \right]
=
0 \, .
\end{equation}
By theorem \ref{thm:D2WkPositivity} this can only happen if $J = K$.
\end{proof}
Let us also record the following corollary which will be of paramount importance.
\begin{corollary}
\label{cor:D2JWkIsContinuouslyInvertible}
For all $k \in \mathbb{R}$ and $J \in \mathcal{S}'_\nu$ the linear map $D^2_J W_k : \mathcal{S}'_\nu \to H (\nu )$ is continuously invertible.
\end{corollary}
\begin{proof}
The bilinear form $\mathcal{S}'_\nu \times \mathcal{S}'_\nu \to \mathbb{R}$ given by
\begin{equation}
\left( K, L \right) \mapsto K \left[ \left( D^2_J W_k \right) \left( L \right) \right]
\end{equation}
is symmetric which can be seen from writing it out explicitly.
By theorem \ref{thm:D2WkPositivity} it is bounded from below.
By continuity $R_\nu^{-1} \circ D^2_J W_k$ is thus self-adjoint, continuous and injective and as such has dense range in $\mathcal{S}'_\nu$.
Since it is also bounded from below it is continuously invertible.
\end{proof}
The surjectivity of $D W_k$ is substantially more involved.
\begin{theorem}
$D W_k : \mathcal{S}'_\nu \to H ( \nu )$ is surjective for all $k \in \mathbb{R}$.
\end{theorem}
\begin{proof}
Let $\phi \in H ( \nu )$.
Then $J \in \mathcal{S}'_\nu$ solves the equation $D_J W_k = \phi$ if and only if for all $K \in \mathcal{S}'_\nu$
\begin{equation}
\label{eq:ExtremisationCondition}
\int_{\mathcal{S}} K \left( \psi - \phi \right) \exp \left[ J \left( \psi \right) \right] f_k \left( \psi \right) \mathrm{d} \nu \left( \psi \right) = 0 \, .
\end{equation}
Since $\phi \in H ( \nu )$ we may apply the Cameron-Martin theorem to obtain the equivalent condition
\begin{equation}
\int_{\mathcal{S}} K \left( \psi \right) \exp \left[ J \left( \psi \right) - \left( R_\nu^{-1} \phi \right) \left( \psi \right) \right] f_k \left( \psi + \phi \right) \mathrm{d} \nu \left( \psi \right)
=
0
\end{equation}
for all $K \in \mathcal{S}'_\nu$.
Let us make the ansatz $J = R_\nu^{-1} \phi + H$ for some $H \in \mathcal{S}'_\nu$.
Then the above is true precisely when $H$ minimises the convex function
\begin{equation}
M_\phi : \mathcal{S}'_\nu \to \mathbb{R} \qquad
T \mapsto \int_{\mathcal{S}} \exp \left[T \left( \psi \right) \right] f_k \left( \psi + \phi \right) \mathrm{d} \nu \left( \psi \right) \, .
\end{equation}
$M_\phi$ is clearly well-defined and continuous because it admits the representation
\begin{equation}
M_\phi \left( T \right) = \exp \left[ \frac{1}{2} \left \langle T, T \right \rangle \right] \int_{\mathcal{S}} f_k \left( \psi + \phi + R_\nu T \right) \mathrm{d} \nu \left( \psi \right)
\end{equation}
in analogy to $Z_k$ in equation \ref{eq:ZRepresentation2}.
We shall now assume that $( H_n )_{n \in \mathbb{N}}$ is a minimising sequence of $M_\phi$ and the goal is to show that there is some bounded subsequence.

Since $f_k$ is continuous and $\mathcal{S}$ admits continuous norms, there is a continuous norm $p$ on $\mathcal{S}$ and some $\delta > 0$ such that
\begin{equation}
\forall \psi \in \mathcal{S} : \quad p \left( \psi \right) \le \delta \implies f_k \left( \phi + \psi \right) \ge \frac{1}{2} f_k \left( \phi \right) \, .
\end{equation}
Now, consider the three mutually exclusive cases
\begin{enumerate}
\item[\underline{$1$}:] $\limsup_{n \to \infty} p \left( R_\nu H_n \right) = 0$,
\item[\underline{$2$}:] there is a subsequence $(K_n)_{n \in \mathbb{N}}$ such that $\lim_{n \to \infty} p \left( R_\nu K_n \right) \in \left( 0, \infty \right)$,
\item[\underline{$3$}:] $\liminf_{n \to \infty} p \left( R_\nu H_n \right) = \infty$.
\end{enumerate}
\underline{$1$}: Suppose there exists a continuous seminorm $q$ on $\mathcal{S}$ such that
\begin{equation}
\limsup_{n \to \infty} \left( p + q \right) \left( R_\nu H_n \right) \neq 0 \, .
\end{equation}
We may then replace the previously used norm $p$ by $p + q$ such that we land in either case \underline{$2$} or \underline{$3$}.
If we have $\limsup_{n \to \infty} ( p + q ) ( R_\nu H_n ) = 0$ for all continuous seminorms $q$ on $\mathcal{S}$, we clearly have that
\begin{equation}
B := \overline{\left \{ R_\nu H_n : n \in \mathbb{N} \right \}} + \left \{ \phi \right \} \subset \mathcal{S}
\end{equation}
is compact.
Also since $\nu$ is Radon there is a compact set $C \subseteq \mathcal{S}$ such that $\nu ( C ) > 0$. 
Clearly,
\begin{equation}
\begin{aligned}
M_\phi \left( H_n \right)
\ge \exp \left[ \frac{1}{2} \left \langle H_n, H_n \right \rangle \right] \int_{C} f_k \left( \psi + \phi + R_\nu H_n \right) \mathrm{d} \nu \left( \psi \right) \, .
\end{aligned}
\end{equation}
Now, $f_k$ is continuous so that it attains its infimum on $B + C$ which cannot be zero.
Hence,
\begin{equation}
\inf_{\psi \in C} f_k ( \psi + \phi + R_\nu H_n )
\ge
\inf_{\psi \in B + C} f_k \left( \psi \right) := \alpha > 0
\end{equation}
and
\begin{equation}
M_\phi \left( H_n \right)
\ge \alpha \exp \left[ \frac{1}{2} \left \langle H_n, H_n \right \rangle \right] \nu \left( C \right) \, .
\end{equation}
But then, $H_n$ cannot be unbounded in $\mathcal{S}'_\nu$ since it is a minimising sequence of $M_\phi$.

For the remaining two cases we shall restrict the integral to a ball of radius $0 < r \le \delta$ in the norm $p$.
Furthermore, $\nu_p$, the pushforward measure of $\nu$ to $\mathcal{S}_p$ via the natural map $\iota_p : \mathcal{S} \to \mathcal{S}_p$ is a Gaußian measure on a Banach space and we obtain
\begin{equation}
\label{eq:MphiRestrictedToBallEstimate}
M_\phi \left( H_n \right)
\ge
\frac{1}{2} f_k \left( \phi \right)
\exp \left[ \frac{1}{2} \left \langle H_n, H_n \right \rangle \right]
\nu_p \left( \overline{B_r \left( - \iota_p R_\nu H_n \right)} \right)
\end{equation}
for all $n \in \mathbb{N}$.
Here, $\overline{B_r ( - \iota_p R_\nu H_n )}$ denotes the closed ball of radius $r$ around $- \iota_p R_\nu H_n$ in $\mathcal{S}_p$.
Note that the transpose $\ltrans{\iota}_p : (\mathcal{S}_p)^* \to \mathcal{S}'_\beta$ has dense range because $\iota_p$ is injective since $p$ is a norm and not just a seminorm (see \cite[Chapter 4, §4, Corollary 2.3]{src:SchaeferWolff:TopologicalVectorSpaces}) and $\mathcal{S}$ is reflexive.
Hence, by the continuity of $M_\phi$ we may assume that $H_n \in \ltrans{\iota}_p (\mathcal{S}_p)^*$ for all $n \in \mathbb{N}$.
Furthermore, letting $R_{\nu_p}$ denote the Hilbert isomorphism between the closure of $(\mathcal{S}_p)^*$ in $L^2(\nu_p)$ and $H(\nu_p)$, it is straightforward to verify that $\iota_p R_\nu \ltrans{\iota}_p$ is equal to the restriction of $R_{\nu_p}$ to $(\mathcal{S}_p)^*$.

\underline{$2$}: Restrict to a subsequence $( \ltrans{\iota}_p K_n )_{n \in \mathbb{N}}$ of $(H_n)_{n \in \mathbb{N}}$ with
\begin{itemize}
\setlength\itemsep{0.5em}
\item $\lim_{n \to \infty} p \left( R_\nu \ltrans{\iota}_p K_n \right) =: P \in (0, \infty)$,
\item $\inf_{n \in \mathbb{N}} p \left( R_\nu \ltrans{\iota}_p K_n \right) > \frac{1}{2} P$,
\item $\sup_{n \in \mathbb{N}} p \left( R_\nu \ltrans{\iota}_p K_n \right) < 2 P$.
\end{itemize}
Set $\gamma = \min \{ \delta, P/2 \}$ and in accordance with \cite[Corollary 7]{src:GaußianMeasureOfShiftedBalls} (considering $t = 1$ only)
\begin{equation}
r = \frac{\gamma}{4} \qquad \text{and} \qquad
\epsilon = 1 - \frac{3}{4} \frac{\gamma}{P} \in \left( 0, 1 \right) \, .
\end{equation}
Then, $r < ( 1 - \epsilon ) p ( R_\nu \ltrans{\iota}_p K_n )$ for all $n \in \mathbb{N}$.
Now, define
\begin{equation}
\begin{aligned}
g_n &= - \left( 1 - \frac{\epsilon \gamma}{8 P} \right) R_{\nu_p} K_n \\
\text{implying}\quad p \left( - R_{\nu_p} K_n - g_n \right)
&=
\frac{1}{8} \epsilon \gamma \frac{p \left( R_\nu \ltrans{\iota}_p K_n \right)}{P}
\le
\frac{1}{4} \epsilon \gamma
=
\epsilon r \, .
\end{aligned}
\end{equation}
and note that $g_n \in R_{\nu_p}(\mathcal{S}_p)^*$ for all $n \in \mathbb{N}$.
Hence, by \cite[Corollary 7]{src:GaußianMeasureOfShiftedBalls},
\begin{equation}
\nu_p \left( \overline{B_r \left( - \iota_p R_\nu K_n \right)} \right)
\ge
\exp \left[ - \frac{1}{2} \left( 1 - \frac{\epsilon \gamma}{8 P} \right)^2 \left \langle \ltrans{\iota}_p K_n, \ltrans{\iota}_p K_n \right \rangle \right]
\nu_p \left( \overline{B_{\left( 1 - \epsilon \right) r} \left( 0 \right)} \right)
\end{equation}
which combines with equation \ref{eq:MphiRestrictedToBallEstimate} to
\begin{equation}
M_\phi \left( \ltrans{\iota}_p K_n \right)
\ge
\frac{1}{2} f_k \left( \phi \right)
\exp \left( \frac{1}{2} \left[ 1 - \left( 1 - \frac{\epsilon \gamma}{8 P} \right)^2 \right] \left \langle \ltrans{\iota}_p K_n, \ltrans{\iota}_p K_n \right \rangle \right)
\nu_p \left( \overline{B_{\left( 1 - \epsilon \right) r} \left( 0 \right)} \right)
\end{equation}
for all $n \in \mathbb{N}$.
By lemma \ref{lem:BallHasNonzeroMeasure} the above is nonzero.
Furthermore, $\epsilon \gamma < 8 P$ such that $( \ltrans{\iota}_p K_n )_{n \in \mathbb{N}}$ must be bounded since it is a minimising sequence of $M_\phi$.

\underline{$3$}: Restrict to a subsequence $( \ltrans{\iota}_p K_n )_{n \in \mathbb{N}}$ of $(H_n)_{n \in \mathbb{N}}$ with $p ( R_\nu \ltrans{\iota}_p K_n ) > 2 \delta$ for all $n \in \mathbb{N}$ and with the same notation as before, set
\begin{equation}
\epsilon = \frac{1}{2}, \qquad r = \delta \qquad \text{and} \qquad g_n = - \left( 1 - \frac{\delta}{2 p \left( R_{\nu_p} K_n \right)} \right) R_{\nu_p} K_n
\end{equation}
for all $n \in \mathbb{N}$.
Then, clearly
\begin{equation}
r < \left( 1 - \epsilon \right) p \left( R_{\nu_p} K_n \right)
\qquad \text{and} \qquad
p \left( - R_{\nu_p} K_n - g_n \right) \le \epsilon r
\end{equation}
such that
\begin{equation}
\begin{aligned}
\nu_p &\left( \overline{B_r \left( - R_\nu \ltrans{\iota}_p K_n \right)} \right) \\
&\ge
\exp \left[ - \frac{1}{2} \left( 1 - \frac{\delta}{2 p \left( R_{\nu_p} K_n \right)} \right)^2 \left \langle \ltrans{\iota}_p K_n, \ltrans{\iota}_p K_n \right \rangle \right] 
\nu_p \left( \overline{B_{\left( 1 - \epsilon \right) r} \left( 0 \right)} \right)
\end{aligned}
\end{equation}
for all $n \in \mathbb{N}$.
Now, note that by \cite[Theorem 3.2.10(i)]{src:Bogachev:GaußianMeasures}, there is a $C > 0$ such that $p ( R_\nu K ) \le C \sqrt{\langle K, K \rangle}$ for all $K \in \mathcal{S}(\mathbb{R}^d)'_\nu$.
Then, since $p ( R_{\nu_p} K_n ) > \delta / 2$ we arrive at
\begin{equation}
\begin{aligned}
\nu_p &\left( \overline{B_r \left( - R_\nu \ltrans{\iota}_p K_n \right)} \right) \\
&\ge
\exp \left[ - \frac{1}{2} \left( 1 - \frac{\delta}{2 C \sqrt{ \left \langle \ltrans{\iota}_p K_n, \ltrans{\iota}_p K_n \right \rangle}} \right)^2 \left \langle \ltrans{\iota}_p K_n, \ltrans{\iota}_p K_n \right \rangle \right] 
\nu_p \left( \overline{B_{\left( 1 - \epsilon \right) r} \left( 0 \right)} \right)
\end{aligned}
\end{equation}
which combines with equation \ref{eq:MphiRestrictedToBallEstimate} to
\begin{equation}
M_\phi \left( \ltrans{\iota}_p K_n \right)
\ge
\frac{1}{2} f_k \left( \phi \right)
\exp \left[ \frac{\delta \sqrt{ \left \langle \ltrans{\iota}_p K_n, \ltrans{\iota}_p K_n \right \rangle}}{4 C} - \frac{\delta^2}{8 C^2} \right]
\nu_p \left( \overline{B_{\left( 1 - \epsilon \right) r} \left( 0 \right)} \right)
\end{equation}
for all $n \in \mathbb{N}$.
As before, since $( \ltrans{\iota}_p K_n )_{n \in \mathbb{N}}$ is a minimising sequence of $M_\phi$ it must be bounded.

Since $M_\phi$ is continuous and convex, it is also weakly lower semicontinuous and hence attains its minimum by the weak compactness of bounded balls in $\mathcal{S}'_\nu$.
\end{proof}
So, $D W_k : \mathcal{S}'_\nu \to H ( \nu )$ is a Fréchet differentiable bijection and by corollary \ref{cor:D2JWkIsContinuouslyInvertible} $D^2_J W_k$ is continuously invertible for all $J \in \mathcal{S}'_\nu$ and $k \in \mathbb{R}$.
\begin{corollary}
\label{cor:DWkInverseIsContinuouslyDifferentiable}
The map $\left( k, \phi \right) \mapsto ( D W_k )^{-1} ( \phi )$ is continuously Fréchet differentiable.
\end{corollary}
\begin{proof}
Define $g : \mathbb{R} \times H ( \nu ) \times \mathcal{S}'_\nu \to H ( \nu )$ with
\begin{equation}
\left( k, \phi, T \right) \mapsto D_T W_k - \phi
=
\frac{D_T Z_k}{Z_k \left( T \right)} - \phi \, .
\end{equation}
It is continuously Fréchet differentiable by lemma \ref{lem:YDifferentiability} and for all $k \in \mathbb{R}, \phi \in H ( \nu )$ and $K, T \in \mathcal{S}'_\nu$
\begin{equation}
\left( D_{k, \phi, T} g \right) \left( 0, 0, K \right) = \left( D^2_T W_k \right) \left( K \right) \, .
\end{equation}
Since $D^2_T W_k$ is continuously invertible, $g$ satisfies the conditions of the implicit function theorem.
\end{proof}

Finally, we may define the \textbf{effective average action} $\Gamma_k : H ( \nu ) \to \mathbb{R}$ as
\begin{equation}
\phi \mapsto \sup_{J \in \mathcal{S}'_\nu} \left[ J \left( \phi \right) - W_k \left( J \right) \right] - \frac{1}{2} F_k \left( \phi, \phi \right) \, .
\end{equation}
for all $k \in \mathbb{R}$.
It is well-defined since the supremum is attained precisely for $J = ( D W_k )^{-1} ( \phi )$.
Hence,
\begin{equation}
\begin{aligned}
\Gamma_k \left( \phi \right) &= \left( D W_k \right)^{-1} \left( \phi \right) \left( \phi \right) - W_k \left( \left( D W_k \right)^{-1} \left( \phi \right) \right) 
- \frac{1}{2} F_k \left( \phi, \phi \right)
\end{aligned}
\end{equation}
for all $\phi \in H ( \nu )$.
By the chain rule the above is also Fréchet differentiable with derivative $D \Gamma_k : H ( \nu ) \to \mathcal{S}'_\nu$ given by
\begin{equation}
\label{eq:DGammakDef}
\phi \mapsto \left( D W_k \right)^{-1} \left( \phi \right) - F_k \left( \phi, \cdot \right)
\end{equation}
where we have used the continuous injection $\mathcal{S}'_\beta \hookrightarrow \mathcal{S}'_\nu$.
This is precisely the \textbf{quantum equation of motion} (see e.g. \cite[Equation 22]{src:Gies:IntroductionToFRG}).
But then, we immediately see that we can take another derivative, leading to
\begin{equation}
D^2 \Gamma_k : H ( \nu ) \to \mathcal{L} \left( H ( \nu ), \mathcal{S}'_\nu \right) \qquad \phi \mapsto \left( D^2_{\left( D W_k \right)^{-1} \left( \phi \right)} W_k \right)^{-1} - F_k \, .
\end{equation}
Thus, the operators $D^2_\phi \Gamma_k + F_k \in \mathcal{L} ( H ( \nu ), \mathcal{S}'_\nu )$ are clearly continuously invertible with the inverses given by $D^2_{( D W_k )^{-1} ( \phi ) } W_k$.

A simple calculation entirely analogous to the standard one (see e.g. \cite{src:Gies:IntroductionToFRG}) now reveals:
\begin{theorem}[The Wetterich equation]
\label{thm:FRGE}
\begin{equation}
\partial_k \Gamma_k \left( \phi \right)
=
\frac{1}{2} \int_{X^k}
\overline{U^k_x} \left[ \left( D^2_\phi \Gamma_k + F_k \right)^{-1} \left( U^k_x \right) \right] \mathrm{d} m_k \left( x \right) + \partial_k \ln N_k
\end{equation}
for all $\phi \in H ( \nu )$.
Here, $\overline{U^k_x} \in ( \mathcal{S}'_\nu )_{\mathbb{C}}$ denotes the complex conjugate of $U^k_x$, i.e with $\overline{U^k_x} ( \phi ) = \overline{U^k_x ( \phi )}$ for all $\phi \in H ( \nu )$.
Note that $\overline{U^k_x}$ is still complex linear since $H ( \nu )$ is a real vector space.
Recall that $X^k$, $m_k$ and $U_x^k$ were defined in equation \ref{eq:FKDerivative}.
\end{theorem}
While this differential equation is in itself already remarkable, its real strength lies in its boundary conditions.
By lemma \ref{cor:DWkInverseIsContinuouslyDifferentiable}, we clearly have
\begin{equation}
\lim_{k \to 0} \Gamma_k \left( \phi \right)
=
\Gamma_0 \left( \phi \right)
=
W_0^c \left( \phi \right) \, ,
\end{equation}
corresponding to the essential ingredient in theorem \ref{thm:ConvergenceTheorem}.
Before we can derive the boundary condition for $k \to \infty$ we need the following theorem.
\begin{theorem}
\label{thm:DiracMeasureApproximation}
Let $g : \mathcal{S} \to \mathbb{R}$ be $\nu$-integrable, continuous at zero and $\vert g \vert \le C \exp [ p^2 ]$ for some $C > 0$ and some continuous seminorm $p$ on $\mathcal{S}$.
Then
\begin{equation}
\lim_{k \to \infty} \frac{\int_{\mathcal{S}} g \left( \psi \right) \exp \left[ - \frac{1}{2} F_k \left( \psi, \psi \right) \right] \mathrm{d} \nu \left( \psi \right)}{
\int_{\mathcal{S}} \exp \left[ - \frac{1}{2} F_k \left( \psi, \psi \right) \right] \mathrm{d} \nu \left( \psi \right)
}
=
g \left( 0 \right) \, .
\end{equation}
\end{theorem}
\begin{proof}
Let us first prove that the measures
\begin{equation}
\theta_k = \frac{\exp \left[ - \frac{1}{2} F_k \left( \cdot, \cdot \right) \right]}{
\int_{\mathcal{S}(\mathbb{R}^d)} \exp \left[ - \frac{1}{2} F_k \left( \cdot, \cdot \right) \right] \mathrm{d} \nu \left( \psi \right)}
\cdot \nu
\end{equation}
converge weakly to the Dirac measure $\delta_0$ at the origin as $k$ goes to infinity.
To that end, let $A_k : \mathcal{S}'_\nu \to \mathbb{R}$ denote the moment-generating function of $\theta_k$.
By the Cameron-Martin theorem we have that
\begin{equation}
\begin{aligned}
A_k \left( T \right)
&=
\exp \left[
T \left( \omega \right)
- \frac{1}{2} \left \langle R_\nu^{-1} \omega, R_\nu^{-1} \omega \right \rangle
-\frac{1}{2} F_k \left( \omega, \omega \right)
\right]
A_k \left( T - \left[ R_\nu^{-1} + F_k \right] \omega \right)
\end{aligned}
\end{equation}
for any $T \in \mathcal{S}'_\nu$ and $\omega \in H(\nu)$.
Here, $F_k \, \omega$ denotes the tempered distribution given by $\phi \to F_k ( \omega, \phi )$.
Now, note that $R_\nu^{-1} + F_k : H(\nu) \to \mathcal{S}'_\nu$ is continuously invertible, by the positivity property of $F_k$ given in equation \ref{itm:FkBound}.
Hence, taking $\omega = ( R_\nu^{-1} + F_k )^{-1} T$, we arrive at
\begin{equation}
A_k \left( T \right)
=
\exp \left[
\frac{1}{2} T \left( \left( R_\nu^{-1} + F_k \right)^{-1} T \right)
\right]
A_k \left( 0 \right)
=
\exp \left[
\frac{1}{2} \left \langle T, \left( \mathrm{id} + F_k R_\nu \right)^{-1} T \right \rangle
\right] \, .
\end{equation}
By analytic continuation $T \mapsto i T$ , we obtain the characteristic functions
\begin{equation}
\hat{\theta}_k \left(T \right) = \exp \left[
- \frac{1}{2} \left \langle T, \left( \mathrm{id} + F_k R_\nu \right)^{-1} T \right \rangle
\right]
\end{equation}
for all $T \in \mathcal{S}'_\beta$.
A necessary condition for the sought convergence is that the functions $\hat{\theta}_k$ converge pointwise to $1$ as $k$ goes to infinity.
Since $F_k ( \phi, \phi )$ is increasing with $k$ for all $\phi \in \mathcal{S}$ by equation \ref{eq:FKDerivative}, $( \mathrm{id} + F_k R_\nu )^{-1}$ is a decreasing family of positive operators on $\mathcal{S}'_\nu$.
Fixing, any $T \in \mathcal{S}'_\nu$ we thus have that $( \mathrm{id} + F_k R_\nu )^{-1} T$ converges in norm to some $K \in \mathcal{S}'_\nu$ as $k$ tends to infinity.
If $K \neq 0$ we obtain
\begin{equation}
\liminf_{k \to \infty} \left \langle K, \frac{\mathrm{id} + F_k R_\nu}{k^2} K \right \rangle
=
\liminf_{k \to \infty} \frac{1}{k^2} F_k \left( R_\nu K, R_\nu K \right)
> 0
\end{equation}
by equation \ref{itm:FkDivergence}.
At the same time
\begin{equation}
\begin{aligned}
\liminf_{k \to \infty} \left \langle K, \frac{\mathrm{id} + F_k R_\nu}{k^2} K \right \rangle
&=
\liminf_{k \to \infty} \left \langle K, \frac{\mathrm{id} + F_k R_\nu}{k^2} \left[ K - \left( \mathrm{id} + F_k R_\nu \right)^{-1} T \right] + \frac{T}{k^2} \right \rangle \\
&=
\liminf_{k \to \infty} \left \langle \frac{\mathrm{id} + F_k R_\nu}{k^2} K, K - \left( \mathrm{id} + F_k R_\nu \right)^{-1} T \right \rangle
= 0 \, ,
\end{aligned}
\end{equation}
where the last equality follows, since $F_k R_\nu / k^2$ is bounded by equation \ref{itm:FkBound}.
Hence, we have a contradiction and may conclude that $\lim_{k \to \infty} ( \mathrm{id} + F_k R_\nu )^{-1} T = 0$ for all $T \in \mathcal{S}'_\nu$.
In particular, we then obtain $\lim_{k \to \infty} \hat{\theta}_k (T ) = 1$ for all $T \in \mathcal{S}'_\beta$.
Now, by \cite[Corollary 3.8.5]{src:Bogachev:GaußianMeasures}, a sufficient criterion for the weak convergence is the uniform tightness of $\{ \theta_k : k \in \mathbb{R} \}$.
To prove it, observe that
\begin{equation}
\int_{\mathcal{S}} T \left( \phi \right)^2 \, \mathrm{d} \nu \left( \phi \right)
=
\left \langle T, T \right \rangle
\ge
\left \langle T, \left( \mathrm{id} + F_k R_\nu \right)^{-1} T \right \rangle
=
\int_{\mathcal{S}} T \left( \phi \right)^2 \, \mathrm{d} \theta_k \left( \phi \right)
\end{equation}
for all $T \in \mathcal{S}'_\beta$ and hence, $\theta_k ( B ) \ge \nu (B)$ for every absolutely convex Borel set $B \subseteq \mathcal{S}$ by \cite[Theorem 3.3.6]{src:Bogachev:GaußianMeasures}.
Because $\nu$ is Radon, for any $\epsilon > 0$ there is a compact set $K \subset \mathcal{S}$ with $\nu (K) > 1 - \epsilon$ which, by the completeness of $\mathcal{S}$, may be taken to be absolutely convex.
Hence, $\theta_k ( K ) > 1 - \epsilon$ as well and we have proven the weak convergence of $\theta_k$ to $\delta_0$.

Let $C$ and $p$ be as stated in the theorem and set $c_k = \theta_k( p^{-1} ( [0, 1 ] ) )$.
Then
\begin{equation}
\liminf_{k \to \infty} c_k
\ge
\liminf_{k \to \infty} \theta_k \left( p^{-1} \left( \left[0, 1 \right) \right) \right)
\ge
1
\end{equation}
by the Portmanteau theorem.
Hence,
\begin{equation}
\liminf_{k \to \infty} \alpha_k
:=
\liminf_{k \to \infty}
\frac{1}{24} \ln \frac{c_k}{1-c_k}
=
\infty
\end{equation}
and there exists $l \in \mathbb{R}$ such that $c_l \ge 3/4$ and $\alpha_l \ge 2$.
Consequently, by Fernique's theorem \cite[Theorem 3.2.10 and Theorem 2.8.5]{src:Bogachev:GaußianMeasures}, $\int_{\mathcal{S}} \exp[2 p^2] \mathrm{d} \theta_l < \infty$ and using \cite[Corollary 3.3.7]{src:Bogachev:GaußianMeasures} as well as the monotonically increasing behaviour of $(F_k)_{k \in \mathbb{R}}$,
\begin{equation}
\sup_{k \ge l} \int_{\mathcal{S}} \exp[2 p^2] \mathrm{d} \theta_k < \infty \, .
\end{equation}
Applying the continuity of $g$ at zero, for every $\epsilon > 0$ there is an open neighbourhood $U \subseteq \mathcal{S}$ of the origin such that $\sup_{\phi \in U} \vert g(\phi) - g(0) \vert < \epsilon$.
Furthermore,
\begin{equation}
\lim_{k \to \infty} \int_{\mathcal{S} \setminus U} \exp[p^2] \mathrm{d} \theta_k
\le
\sqrt{\sup_{k \ge l}  \int_{\mathcal{S}} \exp[2 p^2] \mathrm{d} \theta_k}
\cdot
\limsup_{k \to \infty}
\sqrt{\theta_k \left( \mathcal{S} \setminus U \right)} \\
= 0
\end{equation}
by invoking Hölder's inequality and the Portmanteau theorem again.
Consequently,
\begin{equation}
\begin{aligned}
\limsup_{k \to \infty} &\int_{\mathcal{S}} \left \vert g \left( \phi \right) - g \left( 0 \right) \right \vert \, \mathrm{d} \theta_k \left( \phi \right)
\le
\limsup_{k \to \infty} \int_U \left \vert g \left( \phi \right) - g \left( 0 \right) \right \vert \, \mathrm{d} \theta_k \left( \phi \right) \\
&+ \left[ C + g \left( 0 \right) \right] \limsup_{k \to \infty} \int_{\mathcal{S} \setminus U} \exp \left[ p \left( \phi \right)^2 \right] \, \mathrm{d} \theta_k \left( \phi \right)
\le \epsilon
\end{aligned}
\end{equation}
and the claim follows.
\end{proof}
Furthermore, we have the following useful lemmas.
\begin{lemma}
Let $k \in \mathbb{R}$, $\phi \in H ( \nu )$ and define $X_{k, \phi} : \mathcal{S}'_\nu \to \mathbb{R}$,
\begin{equation}
T \mapsto \frac{1}{N_k} \int_S \exp \left[ T \left( \psi \right)
- S^{\mathrm{int}} \left( \psi + \phi \right)
- \frac{1}{2} F_k \left( \psi, \psi \right) \right]
\mathrm{d} \nu \left( \psi \right) \, .
\end{equation}
Then, $D_\phi \Gamma_k - R_\nu^{-1} ( \phi )$ minimises $X_{k, \phi}$.
\end{lemma}
\begin{proof}
$X_{k, \phi}$ is a moment-generating function and hence convex.
Furthermore, by equations \ref{eq:ExtremisationCondition} and \ref{eq:DGammakDef},
$D_\phi \Gamma_k - R_\nu^{-1} ( \phi )$ extremises $X_{k, \phi}$.
\end{proof}
\begin{lemma}
For all $k \in \mathbb{R}$ and $\phi \in H (\nu )$
\begin{equation}
\exp \left[ - \Gamma_k \left( \phi \right) \right]
=
\exp \left[ -\frac{1}{2} R_\nu^{-1} \left( \phi \right) \left( \phi \right) \right]
X_{k, \phi} \left( D_\phi \Gamma_k - R_\nu^{-1} \left( \phi \right) \right) \, .
\end{equation}
\end{lemma}
\begin{proof}
Let $k \in \mathbb{R}$ and $\phi \in H ( \nu )$.
Then,
\begin{equation}
\begin{aligned}
\exp \left[ - \Gamma_k \left( \phi \right) \right]
&=
\exp \left[ - \left( D W_k \right)^{-1} \left( \phi \right) \left( \phi \right) + W_k \left( \left( D W_k \right)^{-1} \left( \phi \right) \right) + \frac{1}{2} F_k \left( \phi, \phi \right) \right] \\
&=
\frac{1}{N_k}
\int_{\mathcal{S}}
\exp \bigg[
- S^{\mathrm{int}} \left( \psi \right)
+ \left( D W_k \right)^{-1} \left( \phi \right) \left( \psi - \phi \right) \\
&\phantom{= \frac{1}{N_k}
	\int_{\mathcal{S}}
	\exp \bigg[}
+ \frac{1}{2} F_k \left( \phi, \phi \right)
- \frac{1}{2} F_k \left( \psi, \psi \right) 
\bigg]
\mathrm{d} \nu \left( \psi \right) \\
&=
\frac{1}{N_k}
\int_{\mathcal{S}}
\exp \bigg[
- \frac{1}{2} R_\nu^{-1} \left( \phi \right) \left( \phi \right)
- S^{\mathrm{int}} \left( \psi + \phi \right)
+ \left( D W_k \right)^{-1} \left( \phi \right) \left( \psi \right) \\
&\phantom{= \frac{1}{N_k}
	\int_{\mathcal{S}}
	\exp \bigg[}
- F_k \left( \psi, \phi \right) 
- \frac{1}{2} F_k \left( \psi, \psi \right) 
- R_\nu^{-1} \left( \phi \right) \left( \psi \right)
\bigg]
\mathrm{d} \nu \left( \psi \right) \\
&=
\exp \left[ -\frac{1}{2} R_\nu^{-1} \left( \phi \right) \left( \phi \right) \right]
X_{k, \phi} \left( D_\phi \Gamma_k - R_\nu^{-1} \left( \phi \right) \right) \, .
\end{aligned}
\end{equation}
\end{proof}
Collecting all of the above, we can finally calculate the limit for $k \to \infty$.
\begin{theorem}
For all $\phi \in H ( \nu )$,
\begin{equation}
\lim_{k \to \infty} \Gamma_k \left( \phi \right)
=
\frac{1}{2} R_\nu^{-1} \left( \phi \right) \left( \phi \right)
+
S^{\mathrm{int}} \left( \phi \right)
-
S^{\mathrm{int}} \left( 0 \right) \, .
\end{equation}
\end{theorem}
\begin{proof}
Let $k \in \mathbb{R}$ and $\phi \in H ( \nu )$.
Since $D_\phi \Gamma_k- R_\nu^{-1} ( \phi )$ minimises $X_{k, \phi}$, we have
\begin{equation}
\begin{aligned}
\exp \left[ - \Gamma_k \left( \phi \right) \right]
&\le
\exp \left[ -\frac{1}{2} R_\nu^{-1} \left( \phi \right) \left( \phi \right) \right]
X_{k, \phi} \left( 0 \right) \\
&=
\frac{1}{N_k} \int_S \exp \left[
- \frac{1}{2} R_\nu^{-1} \left( \phi \right) \left( \phi \right)
- S^{\mathrm{int}} \left( \psi + \phi \right)
- \frac{1}{2} F_k \left( \psi, \psi \right)
\right]
\mathrm{d} \nu \left( \psi \right) \, .
\end{aligned}
\end{equation}
By assumption, $\exp[-S^{\mathrm{int}}(\phi + \psi)] \le C_\phi \exp [p(\psi)^2]$ for some $C_\phi > 0$ and some continuous seminorm $p$ on $\mathcal{S}$.
Hence, we may now apply theorem \ref{thm:DiracMeasureApproximation} leading to
\begin{equation}
\limsup_{k \to \infty} \, 
\exp \left[ - \Gamma_k \left( \phi \right) \right]
\le
\exp \left[ - \frac{1}{2} R_\nu^{-1} \left( \phi \right) \left( \phi \right)
- S^{\mathrm{int}} \left( \phi \right) + S^{\mathrm{int}} \left( 0 \right) \right] \, .
\end{equation}
For the converse inequality let $n \in \mathbb{N}$ and pick a balanced neighbourhood $U_n$ of zero in $\mathcal{S}$ such that
\begin{equation}
\forall \psi \in U_n : \quad \exp \left[ - S^{\mathrm{int}} \left( \phi + \psi \right) \right] \ge \frac{n}{n + 1} \exp \left[ - S^{\mathrm{int}} \left( \phi \right) \right] \, .
\end{equation}
Then,
\begin{equation}
\begin{aligned}
\inf_{T \in \mathcal{S}'_\nu} X_{k, \phi} \left( T \right)
&\ge
\frac{n}{n + 1} \exp \left[ - S^{\mathrm{int}} \left( \phi \right) \right] \\
&\phantom{\ge} \times
\inf_{T \in \mathcal{S}'_\nu}
\frac{1}{N_k} \int_{U_n} \exp \left[ T \left( \psi \right)
- \frac{1}{2} F_k \left( \psi, \psi \right) \right]
\mathrm{d} \nu \left( \psi \right) \, .
\end{aligned}
\end{equation}
Now, note that the above integral is invariant under the change of $T \mapsto -T$ since $\nu$ is centred.
Furthermore, it is clearly a convex function of $T$.
Hence, the infimum is attained at $T = 0$ and
\begin{equation}
\inf_{T \in \mathcal{S}'_\nu} X_{k, \phi} \left( T \right)
\ge
\frac{n}{n + 1} \exp \left[ - S^{\mathrm{int}} \left( \phi \right) \right]
\frac{1}{N_k} \int_{U_n} \exp \left[ - \frac{1}{2} F_k \left( \psi, \psi \right) \right]
\mathrm{d} \nu \left( \psi \right) \, .
\end{equation}
But then
\begin{equation}
\liminf_{k \to \infty} 
\exp \left[ - \Gamma_k \left( \phi \right) \right]
\ge
\frac{n}{n + 1} \exp \left[ - \frac{1}{2} R_\nu^{-1} \left( \phi \right) \left( \phi \right) - S^{\mathrm{int}} \left( \phi \right)+ S^{\mathrm{int}} \left( 0 \right) \right]
\end{equation}
by theorem \ref{thm:DiracMeasureApproximation} and since $n \in \mathbb{N}$ was arbitrary, the result follows.
\end{proof}
Most commonly, Wetterich's equation is used without regards to domains and often without taking the effect of regularising operators such as $\mathcal{R}$ into account.
Recall that $\mathcal{R}$ (with an extra regularisation index $n$ that we do not write out explicitly in this section) is given by the regularisation scheme introduced in section \ref{sec:FunctionalRegularisationScheme}.
However, we now have the tools to formulate the standard procedure in a rigorous fashion.
The standard choice for $F_k$ is a multiplication operator in Fourier space.
Hence, let us choose $F_k$ such that $X^k = \mathbb{R}^d$ for all $k \in \mathbb{R}$, $\mathcal{A}^k$ is the corresponding Borel sigma algebra and
\begin{equation}
\begin{aligned}
U^k_p &= \mathcal{F}_p := \left[ \phi \mapsto \hat{\phi} \left( p \right) \right] \, , &
R &: \mathbb{R} \times \mathbb{R}^d \to \mathbb{R}, \left( k, p \right) \mapsto R_k \left( p \right) \, , \\
F_k \left( \phi, \phi \right) &= \left \langle \hat{\phi}, R_k \cdot \hat{\phi} \right \rangle_{L^2(\mathbb{R}^d)_{\mathbb{C}}} \, , &
m_k &=  \partial_k R_k \cdot L_d \, ,
\end{aligned}
\end{equation}
where $L_d$ is the Lebesgue measure on $\mathbb{R}^d$, $\partial_k R_k$ is taken to exist $L_d$-almost everywhere and $R$ is regular enough for equation \ref{eq:FKDerivative} to hold.
Let us fix some orthonormal basis $(e_n)_{n \in \mathbb{N}}$ of $L^2(\mathbb{R}^d)$ in $\mathcal{S}(\mathbb{R}^d)$.
With the continuous injection of $\mathcal{S}(\mathbb{R}^d)$ into $\mathcal{S}(\mathbb{R}^d)^*_\beta$ via the $L^2(\mathbb{R}^d)$ inner product, we then have
\begin{equation}
\begin{aligned}
\mathrm{Tr} &\left( D^2_\phi \Gamma_k + F_k \right)^{-1}
=
\sum_{n = 1}^{\infty} \left( D^2_{\left(D W_k \right)^{-1} \left( \phi \right)} W_k \right) \left( e_n, e_n \right) \\
&\le
\frac{2 / N_k}{Z_k \left( \left(D W_k \right)^{-1} \left( \phi \right) \right)}
\int_{\mathcal{S}(\mathbb{R^d})}
\left \Vert \psi \right \Vert_{L^2(\mathbb{R}^d)}^2
\exp \left[ \left(D W_k \right)^{-1} \left( \phi \right) \left( \psi \right) \right]
f_k \left( \psi \right) \mathrm{d} \nu \left( \phi \right)
< \infty
\end{aligned}
\end{equation}
by Hölder's inequality and Fernique's theorem \cite[Theorem 3.2.10]{src:Bogachev:GaußianMeasures}.
In particular, $(D^2_\phi \Gamma_k + F_k)^{-1}$ induces a trace-class operator on $L^2(\mathbb{R}^d)$.
Moreover, by passing to $L^2(\mathbb{R}^d)_{\mathbb{C}}$,
\begin{equation}
\mathrm{Tr} \left( D^2_\phi \Gamma_k + F_k \right)^{-1}
=
\int_{\mathbb{R}^d} \mathcal{F}_{-p} \left[ \left( D^2_\phi \Gamma_k + F_k \right)^{-1} \mathcal{F}_p \right] \mathrm{d} p < \infty \, .
\end{equation}
\begin{remark}
\label{rem:NoTranslationEquivariantRNu}
The same reasoning also applies to $R_\nu$ with
\begin{equation}
\mathrm{Tr} R_\nu
=
\int_{\mathbb{R}^d} \delta_x \left[ R_\nu \delta_x \right] \mathrm{d} x < \infty \, ,
\end{equation}
where $\delta_x$ denotes the Dirac distribution at the point $x \in \mathbb{R}^d$.
Consequently, $R_\nu \neq 0$ cannot be translation equivariant because in that case the above integrand would be constant and the integral infinite.
\end{remark}
Letting $M ( f )$ denote the multiplication operator on $L^2_{\mathbb{C}} ( \mathbb{R}^d )$ by a function $f$, it is now clear that
\begin{equation}
\partial_k \Gamma_k \left( \phi \right)
=
\frac{1}{2} \mathrm{Tr} \left \{
M \left( \partial_k R_k \right) \left[ \mathcal{F} \circ D^2_\phi \Gamma_k \circ \mathcal{F}^{-1} + M \left( R_k \right) \right]^{-1}
\right \}
+
\partial_k \ln N_k \, .
\end{equation}
Apart from the $\partial_k \ln N_k$ term this is exactly the FRGE in physicists' notation.
Making a simple subtraction even removes $\partial_k \ln N_k$ completely.
\begin{theorem}[Wetterich's equation (second formulation)]
\label{thm:FRGE2}
\begin{equation}
\begin{aligned}
\partial_k \bar{\Gamma}_k \left( \phi \right)
=
\frac{1}{2} \mathrm{Tr} \bigg \{
M \left( \partial_k R_k \right) \bigg( &\left[ \mathcal{F} \circ D^2_\phi \bar{\Gamma}_k \circ \mathcal{F}^{-1} + M \left( R_k \right) \right]^{-1} \\
&-
\left[ \mathcal{F} \circ D^2_0 \bar{\Gamma}_k \circ \mathcal{F}^{-1} + M \left( R_k \right) \right]^{-1} \bigg)
\bigg \}
\end{aligned}
\end{equation}
for all $\phi \in H ( \nu )$ where $\bar{\Gamma}_k ( \phi ) = \Gamma_k ( \phi ) - \Gamma_k ( 0 )$.
Furthermore,
\begin{equation}
\Gamma_k \left( 0 \right) = - \inf_{\phi \in H ( \nu )} \bar{\Gamma}_k \left( \phi \right)
\end{equation}
and
\begin{equation}
\lim_{k \to \infty} \bar{\Gamma}_k \left( \phi \right)
=
\frac{1}{2} R_\nu^{-1} \left( \phi \right) \left( \phi \right)
+
S^{\mathrm{int}} \left( \phi \right)
-
S^{\mathrm{int}} \left( 0 \right) \, .
\end{equation}
\end{theorem}
\begin{proof}
The differential equation and the $k \to \infty$ limit is clear.
The expression for $\Gamma_k ( 0 )$ follows from $W_k ( 0 ) = 0$ since $f_k / N_k \cdot \nu$ is a probability measure.
\end{proof}
The subtraction of the $\phi = 0$ contribution of the full propagator is quite remarkable because it is precisely what is done under the hood in concrete calculations involving the \enquote{non-regularised} FRGE.
Furthermore, it motivates the common practise of expanding both sides in powers of $\phi$ and solving the equation order by order.
In fact, it is easy to show that under rather mild conditions one may Gâteaux differentiate under the trace owing to the fact that $Z_k$ is actually Fréchet-$C^\infty$.
Hence, solving order by order is justified but whether the resulting solution is analytic is not clear from these equations.
Resulting expressions including combinatorics may be extracted from \cite{src:Ziebell:ExactFRGPhi4}.

The achievement of theorem \ref{thm:FRGE2} is the rigorous derivation of the FRGE and the exposure of correct domains and boundary conditions of $\bar{\Gamma}_k$.
Notably, the $k \to \infty$ limit depends on $\mathcal{R}$ since $R_\nu^{-1}$ can be seen as the continuous extension of $\ltrans{\mathcal{R}}^{-1} \circ B \circ \mathcal{R}^{-1}$ on a suitable domain.

An unfortunate consequence is that it becomes impossible to choose a trace-class $\mathcal{R}$ in such a way that $R_{\nu}$ is translation equivariant.
Hence, while the differential equation for $\bar{\Gamma}_k$ respects translation invariance the boundary condition at $k \to \infty$ does not which may pose a severe difficulty in concrete calculations.
On the bright side, the fact that the involved operators are trace class and self-adjoint on $L^2_{\mathbb{C}} (\mathbb{R}^d)$ implies that they have a complete basis of eigenvectors.
It seems likely that this fact can be exploited in numerical calculations.
All in all, Wetterich's equation is shown to have a rigorous foundation and if one chooses to parametrise $\mathcal{R}_n$ by cutoffs instead of $\mathbb{N}$, one may compute their corresponding explicit contributions to $\bar{\Gamma}_k$ and $\Gamma_k$ respectively.

Finally, theorem \ref{thm:ConvergenceTheorem} gives a precise condition under which the $n \to \infty$ limit corresponds to a measure on $\mathcal{S}'_\beta$ with $W_n^c = \Gamma_0^n$ now with explicitly written regularisation index.

\section*{Acknowledgements}
I wish to thank Holger Gies, Benjamin Hinrichs, David Hasler, Christoph Kopper, Friedrich Martin Schneider, Jürgen Voigt, Markus Fröb and Stefan Flörchinger for fruitful discussions and helpful advice about this paper.
This work has been funded by the Deutsche Forschungsgemeinschaft (DFG) under Grant Nos. 398579334 (Gi328/9-1) and 406116891 within the Research Training Group RTG 2522/1.
\appendix
\printbibliography
\end{document}